\documentclass[sigconf,nonacm]{acmart}

\usepackage{csquotes}
\usepackage{tabularx}
\usepackage{cleveref}

\makeatletter
\if@ACM@nonacm
  \usepackage{orcidlink}
  \renewcommand\orcid[1]{\unskip\ignorespaces
    \expandafter\gdef\csname typeset@author\the\num@authors\endcsname##1{%
      ##1\,\orcidlink{#1}}}
\fi
\makeatother



\newcommand{\E}{\mathbb{E}}
\newcommand{\PP}{\mathbb{P}}
\newcommand{\Rb}{\mathbb{R}}
\newcommand{\Zb}{\mathbb{Z}}
\DeclareMathOperator{\Var}{Var}
\newcommand{\Fs}{F_S}
\newcommand{\Gst}{G^{*}}
\newcommand{\Lam}{\Lambda_U}
\newcommand{\half}{\tfrac12}
\newcommand{\Hbar}{\bar H}
\newcommand{\eps}{\varepsilon}
\newcommand{\umach}{u_{\mathrm m}}
\newcommand{\tF}{\widetilde F}
\newcommand{\tS}{\widetilde S}
\newcommand{\tG}{\widetilde G}
\newcommand{\tpi}{\widetilde\pi}
\newcommand{\tnu}{\widetilde\nu}
\newcommand{\tLam}{\widetilde{\Lambda}_U}
\newcommand{\tsig}{\widetilde{\sigma}}
\newcommand{\Varone}{\mathrm{Var}_1}
\newcommand{\Gbar}{\bar\Phi}
\DeclareMathOperator{\LSE}{LSE}
\DeclareMathOperator{\poly}{poly}
\newcommand{\Ac}{\mathcal{A}}
\newcommand{\tAc}{\widetilde{\mathcal{A}}}
\newcommand{\Lameta}{\Lambda_\eta}

\AtBeginDocument{%
  \theoremstyle{acmdefinition}%
  \newtheorem{remark}[theorem]{Remark}}

\newcounter{algo}
\newenvironment{algobox}[1]%
  {\par\smallskip\noindent\rule{\linewidth}{0.5pt}\par\nobreak
   \refstepcounter{algo}%
   \noindent\textbf{Algorithm~\thealgo\ --- #1}\par\nobreak\vspace{1pt}\noindent}%
   {\par\vspace{-5pt}\nobreak\noindent\rule{\linewidth}{0.5pt}\par\smallskip}
\Crefname{algo}{Algorithm}{Algorithms}

\begin{document}

\title{Pinpointing Super-Quadratic Quantum Enumeration Speedups}
\subtitle{Exact and Certified Evaluation of the Guessing-Moment Exponent under Product-Distribution Advice}

\author{Carsten Schubert}
\orcid{0009-0004-2795-3650}
\affiliation{%
  \institution{Technische Universität Berlin}
  \department{Electrical Engineering and Computer Science}
  \city{Berlin}
  \country{Germany}
}
\email{carsten.gm.schubert@tu-berlin.de}

\author{Niklas Paskarbeit}
\orcid{0009-0006-6156-9075}
\affiliation{%
  \institution{Freie Universität Berlin}
  \department{Department of Mathematics and Computer Science}
  \city{Berlin}
  \country{Germany}
}

\author{Maximilian J. Kramer}
\orcid{0009-0006-3807-2095}
\affiliation{%
  \institution{Freie Universität Berlin}
  \department{Dahlem Center for Complex Quantum Systems}
  \city{Berlin}
  \country{Germany}
}

\author{Jean-Pierre Seifert}
\orcid{0000-0002-5372-4825}
\affiliation{%
  \institution{Technische Universität Berlin}
  \department{Electrical Engineering and Computer Science}
  \city{Berlin}
  \country{Germany}
}

\author{Marian Margraf}
\orcid{0009-0005-8577-1318}
\affiliation{%
  \institution{Freie Universität Berlin}
  \department{Department of Mathematics and Computer Science}
  \city{Berlin}
  \country{Germany}
}

\renewcommand{\shortauthors}{Schubert et al.}

\begin{abstract}
	In cryptanalysis and beyond, presumably the most basic problem is search within an unstructured space: Given a finite set of elements $K$, find a marked element using an efficiently evaluable predicate.
    Grover's algorithm \cite{groverFastQuantumMechanical1996} gives an optimal quadratic query advantage for
    black-box search.
    Cryptanalytic settings, however, often come with additional probabilistic advice over the candidate space.
    These distributions frequently have product form, e.g.\@ when probabilities for independent key coordinates are derived from side-channel leakage.
	Classically, simply querying the predicate function in the given likelihood order is optimal to find the correct result as fast as possible in expectation. In the quantum setting, \citet{montanaroQuantumSearchAdvice2011} showed how to achieve an optimal expected query complexity, beating plain Grover on every non-uniform distribution of the advice (disregarding a constant overhead factor).
    What has been missing so far is a finite-size method for evaluating the quantum--classical guessing-moment separation induced by a given advice distribution. We provide such a method for product-distribution advice, thereby sharpening the previous entropy-based estimate of \citet{bashiri_super-quadratic_2026}.
    Our approach reduces the relevant classical and quantum guessing moments to functionals of the one-dimensional \emph{surprisal} distribution. For product advice, this distribution is obtained by convolving the per-coordinate surprisal laws. When the surprisals lie on a common arithmetic grid---the \emph{commensurate case}---the resulting logarithmic moments, and hence the speedup exponent, can be evaluated as finite sums without discretization error. Exponential tilting makes this computation numerically stable. For general non-commensurate product advice, we discretize the surprisals onto a common grid before running our algorithm and derive an a-posteriori bound on the resulting binning error.
    We apply the framework to cold-boot leakage on seeds and block-cipher keys, to template-attack posteriors, and to synthetic i.i.d.\@ Bernoulli posteriors whose parameter is chosen to reproduce average residual ranks reported for Keccak side-channel attacks on ML-KEM and ML-DSA. The resulting guessing-moment exponents substantially exceed $2$ in several skewed-advice settings, reaching $3.92$ for cold-boot advice and $3.12$ in the synthetic models, and include cases where the previous entropy-based bound did not establish an exponent above $2$.
\end{abstract}

\maketitle

{\raggedright\textcolor{red}{This is an extended version of the paper in ACM PQQS~'26, \href{https://doi.org/10.1145/3839338.3857257}{doi:10.1145/3839338.3857257}.}\par}

\section{Introduction}
The most basic task in cryptanalysis is \emph{key guessing}: recover a secret
$x$ given only the ability to test candidates.\footnote{In general, these enumeration and search methods are of course not limited to cryptanalytic settings. But in order to make the formulation more consistent with such typical applications, we stick to their terminology from now on.}
When nothing at all is known about $x$,
the best a quantum attacker can do is
Grover search~\cite{groverFastQuantumMechanical1996}: quadratically fewer guesses than any classical enumeration, and provably no better in this black-box setting.
In practice there often is additional information that can be algorithmically exploited by an attacker.
Passwords are a familiar example: a capital letter is most likely at the first position,
digits and punctuation cluster at the end, and in-between the letters of
ordinary language appear at different rates (letter \texttt{e} is more
frequent than letter \texttt{y}). If~an attacker knows only these positional
frequencies, and not how any character influences the next, the
natural model is a \emph{product distribution}: one
marginal $p_i$ per position $i \in \{1, \ldots, m\}$, multiplied together.
The~resulting distribution $P$
over the key space is the attacker's \emph{advice}.
Similar product shapes
arise well beyond this deliberately simplified scenario.
Further cryptologic examples include independent
per-coordinate side-channel leakage \cite{chari_template_2003,martinQuantumKeySearch2018} or secrets sampled coordinatewise \cite{regev_lattices_2009,bos_crystals_2018}.
Equipped with such advice, the optimal classical attacker
iteratively computes an ordering from most to least likely keys
and guesses along that list.
We define the \emph{rank} $G(x)$ of key $x$ as its position in that list, so reaching $x$ costs $G(x)$ guesses.
Thus, the expected cost of the classical attack is the average rank $\E_P[G]$, the first \emph{guessing moment},
and it can be far smaller than the size of the key space.
We study the resulting quantum--classical separation in expected query complexity. Specifically, we compute the guessing-moment exponent obtained from the optimal classical moment $\E_P[G]$ and the quantum-relevant moment $\E_P[\sqrt{G}]$, including finite-size and tie effects.

The quantum counterpart of \emph{search with advice} is a more elaborate version of Grover search, due to
\citet{montanaroQuantumSearchAdvice2011}. Instead of searching the whole key space at once,
the algorithm walks down the ordered candidate list in segments of geometrically increasing
size and runs a Grover search on each one, so that a likely key is
found in one of the small early segments and never costs a full search.
Reaching a key of rank $G(x)$ this way takes $\Theta(\sqrt{G(x)})$ queries, so the expected query complexity is $\Theta(\E_P[\sqrt{G}])$.
Montanaro proved this dependence optimal up to constant factors.
Writing $f(\rho)\coloneq  \ln\E_P[G^\rho]$ for the logarithmic guessing moments, we define the \emph{speedup exponent} $s\coloneq f(1)/f(1/2)$. Equivalently, $\E_P[\sqrt{G}] = \E_P[G]^{1/s}$, so the moment governing quantum search is the $1/s$-th power of the
classical expected cost. Montanaro's bounds transfer this relation to the optimal quantum query complexity up to universal constant factors and an additive overhead; see \Cref{sec:setup}.
Jensen's inequality gives
$f(\half)\le \half f(1)$ and hence $s\ge2$. The Grover-like baseline approaches $s=2$
as the problem size grows. Thus $s$ is a finite-size guessing-moment
exponent, while a limiting value bounded away from $2$ corresponds
to a genuinely \emph{super-quadratic separation}.
What has been missing so far is a way to \emph{compute}
the speedup exponent $s$ for a given advice distribution, i.e.\@ to tell, at any concrete scheme, how far
above $2$ that exponent exactly sits.

A first advance in that direction is due to
\citet{bashiri_super-quadratic_2026}, who applied \citeauthor{arikanInequalityGuessingIts1996}'s guessing
inequality~\cite{arikanInequalityGuessingIts1996} to bound both guessing moments in terms of R\'enyi
entropies of the advice.
Their estimate bounds $s$ from below and is tight only asymptotically, for large $m$.
At~the finite parameters cryptology cares about, this route cannot be sharpened:
Arıkan's bounds pin each moment only to within an \emph{additive}
slack, and those two slacks, unrelated to one another, leave an interval on the
ratio $s$ that does not close.
What an attacker cares about, however, is the surplus $s-2$ over the
Grover-like regime, and the unresolved interval can hide a substantial
fraction of that. In some parameter settings, the bound from \cite{bashiri_super-quadratic_2026} no longer certifies
$s>2$ at all (see \Cref{sec:apps}).

In this paper we not only bound $s$, we also present a method to compute it.
Our sole structural assumption is the product form met above, i.e.\@ $P=\prod_{i=1}^m p_i$.
For arbitrary product advice, the binned computation is accompanied by a two-sided a-posteriori certificate on its discretization error.
The~certificate is linear in the bin width; for bounded alphabets, certifying a target accuracy $\delta$ therefore costs $\widetilde O(m\,C(s)/\delta)$ (\Cref{thm:precision}, \Cref{thm:bin}).\footnote{Throughout, $\widetilde O$ suppresses polylogarithmic factors.
}
Our main approach is based on one fundamental observation, and a further reduction makes it generally applicable.
The observation is that guessing in order of decreasing probability is the same
as guessing in order of increasing \emph{surprisal} $S(x)=-\ln P(x)$.
Thus, a key's rank, and with it both guessing moments, is governed by the distribution
of this single quantity (\Cref{lem:rankcdf}). Because $P$ is a product, the
surprisal of a key is a \emph{sum} of independent per-coordinate contributions,
and the $m$-dimensional problem collapses to a one-dimensional one. When the
per-coordinate surprisals share a common arithmetic grid---we call this the
\enquote{commensurate} case---that sum is a finite computation and $s$ comes
out \emph{exactly} with no discretization error (\Cref{thm:exact}).

Real-world scenarios are usually not commensurate: letter frequencies, for instance, are
arbitrary real numbers and share no grid. The subsequent reduction handles every such case
in a single step: round the per-coordinate surprisals to a grid of spacing
$\eta$, forcing commensurability, then run the previously established computation.
Under the scaling-law hypotheses of \Cref{thm:bin}, the certified bound on the error this introduces shrinks in proportion to $\eta$, so halving the spacing also halves that bound.
Beyond that scaling law, every run yields an a-posteriori bound on the binning error of the exact estimator, free of unspecified constants or asymptotic ingredients (\Cref{prop:cert}).
One mathematical apparatus (exponential tilting), developed once, is used for the exact and general cases
and the error analysis alike.

\paragraph{Contributions and roadmap.}
\Cref{sec:setup} fixes the guessing model, the exponent and the standing
assumptions. In \Cref{sec:rankcdf}, we reduce both guessing moments to
functionals of a single one-dimensional law (illustrated on a toy advice in \Cref{fig:example}), turning the speedup
exponent from \emph{something to be estimated} into \emph{something computable}.
In \Cref{sec:engine}, we give an exact method for commensurate advice, evaluating
$s$ as a finite sum with no discretization error (\Cref{thm:exact}), and
realize it through exponential tilting so that the exponentially small
quantities involved remain representable in ordinary double precision
(\Cref{lem:tiltconv,lem:fftcost,lem:stable}). In \Cref{sec:binning}, we lift
the commensurability restriction, binning arbitrary advice onto a common grid
and bounding both the error incurred and the cost of driving it down
(\Cref{thm:precision}), which brings every product distribution within reach.
In~\Cref{sec:accuracy}, we prove the scaling law bounding that error and
construct the certificate accompanying each run (\Cref{prop:cert}), so that a
reported exponent comes with a non-asymptotic a-posteriori bound on its
discretization error rather than an asymptotic promise. Finally, in
\Cref{sec:apps}, we instantiate the framework on cryptographic advice settings where
it is meaningful
and read off both the speedup's size and the earlier bound's gap. All
results are stated in the body, which is self-contained; the appendices carry proofs and supporting material.

\paragraph{Related work.} Montanaro's routine \cite{montanaroQuantumSearchAdvice2011} assumes its advice is already
presented in likelihood order, which is not immediate at cryptographic key sizes.
\citet{martinQuantumKeySearch2018} gave an algorithm that essentially removed that assumption
by returning the $r$-th key of a weight band efficiently.
Later analyses, ours included, assume \citeauthor{martinQuantumKeySearch2018}'s method; their work also first explicitly proposed Montanaro's quantum search for side-channel attacks.
Subsequent variants of quantum search with advice generalize the setting or sharpen overhead constants:
\citet{he_quantum_2024} maximize success probability at a
fixed query budget, \citet{albrecht_quantum_2023} admit bounded-error oracles,
and \citet{anderson_improved_2024} carry the advice idea to span programs.
None of them moves $s$ as we define it, since Montanaro's two-sided bracket pins it to the
problem rather than to any particular routine (see \Cref{sec:setup}). Apart from \cite{bashiri_super-quadratic_2026},
\citet{budroni_further_2024} are thematically closest to us: they exactly compute classical and quantum query complexity under the stronger assumption that all product marginals are identical.
On the defensive side,
\citeauthor{fischlinTighterBitSecurityBounds2026}~\cite{fischlinBitSecurityQuantum2026,fischlinTighterBitSecurityBounds2026}
show that keys within statistical
distance $2^{-\lambda/2}$ of uniform retain at least $\lambda/2$ bits of security.
The moments themselves predate the
quantum question: \citet{massey_guessing_1994} showed that Shannon entropy yields no
upper bound on $\E_P[G]$, \citet{arikanInequalityGuessingIts1996} supplied the two-sided
R\'enyi bracket, and \citet{rioul_variations_2022} develops the full family
$\E_P[G^\rho]$ from which $f(1)$ and $f(\half)$ are drawn.

\section{Setup and the speedup exponent}\label{sec:setup}
\paragraph{The guessing problem.} Let $K=\prod_{i=1}^m A_i$ be a product key
space with $|A_i|=b_i$ and $N\coloneq |K|=\prod_i b_i$, and let the advice be a
product distribution $P(x)=\prod_i p_i(x_i)$; write $p_{i,a}\coloneq p_i(a)$. The
\emph{surprisal} of a key decomposes coordinatewise,
$S(x)=-\ln P(x)=\sum_{i=1}^m s_i(x_i)$ with $s_i(a)\coloneq-\ln p_{i,a}$. A \emph{guessing order} assigns a distinct position
$G(x)\in\{1,\dots,N\}$ to each key, and its cost enters through the moments
$\E_P[G^\rho]$, for whose logarithms we write $f(\rho)\coloneq \ln\E_P[G^\rho]$.
Enumerating by increasing surprisal minimizes these moments
simultaneously for every $\rho>0$, and every such order attains the same value
(\Cref{lem:rankcdf}). Keys of equal surprisal may therefore be enumerated in
any order, and no tie convention enters the moments. We write $G$ for any
optimal order and reserve $\Gst(x)\coloneq\#\{x'\in K:S(x')\le S(x)\}$ for the \emph{block-maximal} rank, the largest rank in $x$'s equal-surprisal block. Thus $G(x)\le\Gst(x)$; equality holds for the key placed last within its tie block, and for every key on that level iff the block is a singleton.
Logarithms are natural throughout --- surprisals, entropies, the
bin width $\eta$ and all saddle quantities are in nats --- while $s$, being a
ratio of logarithms, is base-free.

\paragraph{Cost model and the exponent.} By $W_C$ and $W_Q$ we denote the optimal
classical and quantum query complexities of guessing under advice $P$.
Classically, $W_C=\E_P[G]$ exactly. Montanaro's bounds~\cite{montanaroQuantumSearchAdvice2011} show that, up to universal constant factors, the quantum complexity is governed by the square-root guessing moment, $W_Q=\Theta(\E_P[\sqrt G])$; the
algorithm attaining the upper bound additionally incurs a lower-order additive
overhead.
Consequently, $\ln W_C=f(1)$ and $\ln W_Q=f(\half)+O(1)$. We therefore define the information-theoretic \emph{speedup exponent} $s\coloneq f(1)/f(\half)$.
The corresponding operational logarithmic exponent satisfies
$\frac{\ln W_C}{\ln W_Q}=s\left(1+O(1/f(\half))\right)$,
and hence approaches $s$ whenever $f(\half)\to\infty$; for families with $f(\half)=\Theta(m)$ the relative discrepancy is $O(1/m)$. All subsequent statements concern this guessing-moment exponent $s$.
Even~with the operational overhead set aside, $s$ exceeds $2$ at every \emph{finite} $N$, even for uniform advice; unlike non-uniform advice, however, uniform advice has $s\to2$ \emph{asymptotically}, as $N\to\infty$.

\paragraph{Information-theoretic quantities.}
Two~distributions on $K$ are in play: a key's rank counts how many keys are at least as likely, which refers to the uniform distribution $U$ on $K$, while the guessing moments average under the advice $P$.
Accordingly, the surprisal $S$ has two laws (\emph{law} being synonymous with \emph{distribution} here), one under $U$ and one under $P$, both introduced in \Cref{sec:rankcdf}.
For $\alpha\in(0,1)$, the R\'enyi entropy $H_\alpha(P)=\frac1{1-\alpha}\ln\sum_xP(x)^\alpha$ is additive for product laws, it has the Shannon entropy $H_1(P)\coloneq-\sum_xP(x)\ln P(x)$ as its continuous extension. Write $\Hbar_\alpha\coloneq\frac1m\sum_iH_\alpha(p_i)$ and $H\coloneq\E_P[S]=H_1(P)=m\Hbar_1$.

The~mean surprisal under $U$ is $\mu_U\coloneq\E_U[S]$.
Arıkan's guessing inequality \cite{arikanInequalityGuessingIts1996} gives
\begin{align}\label{eq:arikanscale}
	f(\rho)=\rho\sum_{i=1}^mH_{1/(1+\rho)}(p_i)
	+O\left(\log(1+\ln N)\right).
\end{align}
For uniformly bounded alphabets this yields
$f(\rho)=\rho m\Hbar_{1/(1+\rho)}(1+O(\log m/m))$
whenever the per-coordinate entropy stays bounded away from zero, and hence $s=2\Hbar_{1/2}/\Hbar_{2/3}\,(1+o(1))$.
This recovers the asymptotic entropy scale underlying \cite{bashiri_super-quadratic_2026}; see \Cref{app:cscale}.

\paragraph{Standing assumptions.} We assume throughout:

\begin{itemize}
	\item[\emph{(i)}] $P$ is a
product distribution, which is what makes the surprisal additive and the
whole analysis one-dimensional; advice with genuine dependence between
coordinates is out of scope, though richer models often fit by grouping
positions into coarser coordinates and taking their product.
\item[\emph{(ii)}]
$p_{i,a}>0$ for every $i,a$ is a convention rather than a restriction:
symbols of zero probability can simply be removed from $A_i$.
\item[\emph{(iii)}] Complexity is measured in the
input size $B_{\mathrm{in}}\coloneq \sum_ib_i$ and, where discretization enters, the
level count; no uniform bound on the alphabet sizes is assumed.
\item[\emph{(iv)}]
For the saddle-point analysis we assume non-degeneracy, $\Var_U(S)>0$, equivalently that $P$ is not uniform. The uniform case is excluded only from that analysis and can be evaluated directly.
\end{itemize}
\paragraph{Notation.} \Cref{tab:notation} collects the recurring symbols.
\begin{table}[t]\scriptsize
	\caption{Recurring notation.}\label{tab:notation}
	\begin{tabularx}{\columnwidth}{@{}l@{\hskip 6pt}>{\raggedright\arraybackslash}X@{}}
		\toprule
		symbol                                 & meaning                                                                                                        \\
		\midrule
		$K,\ A_i,\ b_i,\ m,\ N$                & key space $K=\prod_{i=1}^m A_i$; alphabets, $|A_i|=b_i$; $N=|K|$ keys; input size $B_{\mathrm{in}}=\sum_ib_i$  \\
		$P=\prod_i p_i,\ U$                    & product posterior (advice), marginals $p_i$; uniform law on $K$                                                \\
		$s_i(a)=-\ln p_{i,a}$                  & per-coordinate surprisal                                                                                       \\
		$S=\sum_i s_i$                         & total surprisal                                                                                                \\
		$\pi,\ \nu;\ \Fs$                      & laws of $S$ under $P,U$ (mass, count); uniform CDF $\Fs(t)=\PP_U[S\le t]$ \\
		$\pi_i,\ \nu_i;\ n$                    & per-coordinate mass and count laws; level size $n=N\nu$                                                        \\
		$G;\ \Gst=N\cdot\Fs(S)$                     & optimal guessing order; block-maximal rank (\Cref{lem:rankcdf})                                                \\
		$\Ac_\rho(t);\ \tAc_\rho$              & level mean power rank $n(t)^{-1}\sum_{r\in J_t}r^\rho$ (\Cref{lem:rankcdf}); binned analogue                   \\
		$f(\rho),\ s$                          & log guessing moment $\ln\E_P[G^\rho]$; exponent $s=f(1)/f(\half)$                                              \\
		$H_\alpha,\ \Hbar_\alpha$              & R\'enyi entropy and its per-coordinate mean $\frac1m\sum_iH_\alpha(p_i)$                                       \\
		$H$                                    & Shannon entropy $H=\E_P[S]=m\Hbar_1$                                                                           \\
		$\mu_U;\ \Lam(\tau)$                   & mean of $S$ under $U$; uniform CGF $\ln\E_U[e^{\tau S}]$, so $\Lam'(0)=\mu_U$                                  \\
		$U_\tau$                               & tilted law $U_\tau^{(i)}\propto e^{\tau s_i}$, $U_{-1}=P$                                                       \\
		$\tau_t,\ \sigma_t$                    & saddle $\Lam'(\tau_t)=t$, scale $\sigma_t^2\coloneq\Lam''(\tau_t)$                                              \\
		$\theta^\star$                         & edge rank-growth rate $\theta^\star=|\tau_H|=1$ (\Cref{cor:rate})                                               \\
		$\phi,\ \mu^{(\phi)};\ \phi_\rho$      & engine tilt and tilted grid law; working tilt $\phi_\rho=-1/(1+\rho)$ (\Cref{sec:stable}) \\
		$\eta,\ h;\ R;\ \eps_0$                & bin width / lattice span; surprisal range; scaling-law regime constant                        \\
		$\Lameta$                              & lattice-span exponent \eqref{eq:lameta}                                                                         \\
		$\widetilde s_i,\tS,\tpi,\tnu,\tF,\tG$ & binned analogues (also $\tpi_i,\tnu_i$): strictly upward rounding to $\eta\Zb$, $0<\tS-S\le m\eta$                                    \\
		$\widetilde U_\tau,\ \tLam;\ \widetilde\tau_t,\ \tsig_\tau$ & binned tilt $\widetilde U_\tau^{(i)}\propto e^{\tau\widetilde s_i}$, its CGF $\tLam(\tau)=\ln\E_U[e^{\tau\tS}]$; binned saddle $\tLam'(\widetilde\tau_t)=t$, $\tsig_\tau^2=\tLam''(\tau)$ (\Cref{lem:lclt}) \\
		$f_\eta,\ s_\eta;\ \widehat s$         & binned moments, exponent; floating-point output (\Cref{lem:stable})                                            \\
		$\Delta(x)$                            & binning window log-ratio, floored at rank $1$ (\Cref{lem:sandwich})                                            \\
		$C(s),\ \eps_m$                        & scaling-law constant $(2+s)/\Hbar_{2/3}$; finite-size correction of \Cref{thm:bin}\\
		$\eps_\tau,\ \beta_\tau$               & local-CLT error, summed tilted third moment (\Cref{lem:lclt})                                                  \\
		$\delta,\ B;\ \umach,\ \poly$          & target precision; a-posteriori error bound (\Cref{prop:cert}); unit roundoff; polynomial in $m,1/\eta$\\
		\bottomrule
	\end{tabularx}
\end{table}

\section{The rank--CDF reduction}\label{sec:rankcdf}
\begin{figure}[t]
	\centering
	\includegraphics[width=\columnwidth]{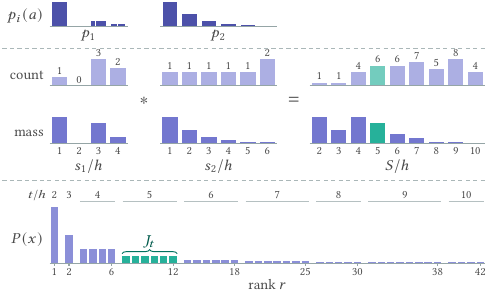}
	\caption{Running example: the input marginals (top), the convolution of their count and mass laws (middle), and their $N=42$ keys in likelihood order (bottom).}
	\Description{Three rows of bar charts. Top: the symbol probabilities of two coordinates, placed above their surprisals. Middle: count and mass histograms of both coordinates over surprisal and their convolution. Bottom: the probabilities of all 42 keys in likelihood order, grouped into nine levels.}
	\label{fig:example}
\end{figure}
The guessing moments average over all $N$ keys, where $N$ can be exponential in $m$, so direct enumeration is infeasible.
The next lemma removes the key space from the picture, replacing it with two
one-dimensional distributions of the surprisal: how many keys sit at each
surprisal level, and how likely the secret is to sit there.
Because $P$ is a product, each is an $m$-fold convolution of $m$ small per-coordinate
distributions. Every subsequent algorithm in this paper computes those convolutions.

The approach works because the surprisal of a key, while it does not fix that
key's rank, confines it narrowly: keys of equal surprisal occupy consecutive
ranks, and the moment averages over those ranks exactly. Individual ranks
therefore depend on how ties are broken, but the moments do not.

Fix a surprisal level $t$. Counting keys is an operation under the uniform
distribution $U$, so we write $\Fs(t)\coloneq\PP_U[S\le t]$, $\Fs(t^-)\coloneq\PP_U[S<t]$ and $\nu(t)\coloneq\PP_U[S=t]$
for the cumulative distribution function of $S$ under $U$, its left limit, and
its mass at $t$. The first of these is what a rank measures: $N\cdot\Fs(t)$ is
exactly the number of keys at least as likely as one of surprisal $t$, and
$n(t)\coloneq N\cdot\nu(t)$ is the number of keys at that level. Accordingly we set
\[
J_t\coloneq
\{N\Fs(t^-)+1,\ldots,N\Fs(t)\}, \text{~~and~~}
\Ac_\rho(t)\coloneq\frac1{n(t)}\sum_{r\in J_t}r^\rho,
\]
an integer interval containing exactly $n(t)$ ranks and the mean $\rho$-th
power of a rank across it.
The~advice enters at one point only, through $\pi(t)\coloneq \PP_P[S=t]$, the probability that the secret has surprisal $t$.
Since every key at level $t$ has probability $e^{-t}$, mass is count times weight:
\begin{equation}\label{eq:pinu}
	\pi(t)=n(t)\,e^{-t}=N\,\nu(t)\,e^{-t}.
\end{equation}
Summing \eqref{eq:pinu} over levels $u\le t$ and using $\sum_{u\le t}\pi(u)\le1$ gives an exact Chernoff-type bound we will reuse:
\begin{equation}\label{eq:rankchernoff}
	N\cdot\Fs(t)=\sum_{u\le t}e^{u}\,\pi(u)\ \le\ e^{t},
	\qquad\text{i.e.}\qquad \Gst\le e^{S}\ \text{pointwise};
\end{equation}
a key's rank never exceeds one over its probability, with near-equality at typical keys, where $S\approx H$.

\begin{lemma}[rank--CDF reduction]\label[lemma]{lem:rankcdf}
	\phantom{-----------}
	\begin{itemize}
		\item[(i)] Among guessing orders, each moment $\E_P[G^\rho]$ is minimized
		   		exactly by the likelihood-ordered
		      ones (for $\rho>0$), and the value is invariant under permutations within
		      equal-surprisal blocks (proof in \Cref{app:stable}).
		\item[(ii)] Every optimal $G$ satisfies
		      $N\cdot\Fs(S(x)^-)<G(x)\le N\cdot\Fs(S(x))=\Gst(x)$; that is, the ranks on
		      level $t$ fill $J_t$.
		\item[(iii)] For every $\rho>0$, $\E_P[G^\rho]=\sum_t\pi(t)\,\Ac_\rho(t)$ and therefore $f(\rho)=\ln\sum_t\pi(t)\,\Ac_\rho(t)$.
	\end{itemize}
\end{lemma}
\begin{proof}
	(ii) Since $t\mapsto e^{-t}$ is strictly decreasing, a likelihood order ranks
	every key of surprisal $<t$ before, and every key of surprisal $>t$ after,
	the $n(t)$ keys at level $t$; those keys therefore occupy exactly the ranks
	of $J_t$.
	(iii) On $\{S=t\}$ every key has probability $e^{-t}$, so
	$\sum_{x:S(x)=t}P(x)\,G(x)^\rho=e^{-t}\sum_{r\in J_t}r^\rho=\pi(t)\,\Ac_\rho(t)$;
	sum over levels. By (i) the value does not depend on which optimal $G$ was
	chosen.
\end{proof}

\Cref{fig:example} traces these objects on a toy advice.
Its~two marginals,
\begin{align*}
	p_1&=(2^{-1},2^{-3},2^{-3},2^{-3},2^{-4},2^{-4}),\\
	p_2&=(2^{-1},2^{-2},2^{-3},2^{-4},2^{-5},2^{-6},2^{-6}),
\end{align*}
are dyadic, so every surprisal is a multiple of $h=\ln2$, and their $N=6\cdot7=42$ keys fall into nine levels.
The~highlighted level $t=5h$ holds $n(t)=6$ keys of probability $2^{-5}$, which fill $J_t=\{7,\dots,12\}$ and carry $\pi(t)=6\cdot2^{-5}$.
Summing over the levels as in \Cref{lem:rankcdf}~(iii) gives $\E_P[G]=3741/512\approx7.31$ and $\E_P[\sqrt G]\approx2.37$, hence $s\approx2.30$, whereas uniform advice on $42$ keys gives $s\approx2.07$.

Item (iii) is the identity every later algorithm evaluates. What remains is to
compute the two surprisal distributions it involves:
without discretization error when the per-coordinate surprisals are commensurate (\Cref{sec:engine}), and otherwise through the binned approximation and its a-posteriori error bound (\Cref{sec:binning,sec:accuracy}).

\section{Exact evaluation and the stable engine}\label{sec:engine}
This section describes the main computation engine.
For commensurate marginals, $s$ is given \emph{exactly} by a finite sum, with no discretization error (\Cref{sec:exact}); \Cref{sec:stable} gives its numerically stable floating-point realization for typical $m$.
Throughout this and the next section, $\poly$ denotes a fixed polynomial in
$m,1/\eta$, and $h$ may be read for $\eta$ in the commensurate case. Cost
statements are in \emph{word} operations at fixed precision, which
\Cref{lem:stable} justifies by showing double precision suffices for the
targets of interest; under the interval-arithmetic route the working precision
becomes a parameter and the bit complexity scales with it.

\subsection{The exact evaluation}\label{sec:exact}
Say the marginals are \emph{commensurate} if all $\ln p_{i,a}$ lie on a common
arithmetic grid, $s_i(a)\in h\Zb$ for some span $h>0$. Then $S$ takes values in
$h\Zb$, and since it ranges over an interval of width
$R\coloneq \sum_i(\max_as_i(a)-\min_as_i(a))$, the \emph{surprisal range}, its
support has $O(R/h)$ points (polynomial when the range is $O(m)$,
as under bounded per-coordinate leakage). Dyadic models --- and more generally
all models whose probabilities are integer powers of a common base,
$p_{i,a}=e^{-k_{i,a}h}$, so that their surprisals are specified in fixed
units --- are commensurate by construction.

On such a grid \Cref{lem:rankcdf} (iii) becomes a finite computation: the level
sum runs over the $O(R/h)$ grid points, and $\Ac_\rho(t)$ depends on the level
only through the pair $(\Fs(t^-),\Fs(t))$.
The~two laws of \Cref{sec:rankcdf} are convolutions of explicit per-coordinate laws, the \emph{count laws} $\nu_i(w)\coloneq\#\{a:s_i(a)=w\}/b_i$ and the \emph{mass laws} $\pi_i(w)\coloneq\sum_{a:s_i(a)=w}p_{i,a}$, namely $\nu=\nu_1*\cdots*\nu_m$ and $\pi=\pi_1*\cdots*\pi_m$.
By~\eqref{eq:pinu}, only the first convolution has to be computed, and $\pi$ follows from it level by level.
In~\Cref{fig:example}, where $h=\ln2$ and $z^k$ marks surprisal $kh$, the level sizes $n=N\nu$ are the coefficients of a product of count polynomials,
$(z+3z^3+2z^4)(z+z^2+z^3+z^4+z^5+2z^6)=z^2+z^3+4z^4+6z^5+6z^6+7z^7+5z^8+8z^9+4z^{10}$.

\begin{theorem}[exact evaluation]\label[theorem]{thm:exact}
	In the commensurate case
	\[
		f(\rho)=\rho\ln N+\ln\sum_t\pi(t)\,\frac{\Ac_\rho(t)}{N^\rho}
	\]
	is a finite sum over the grid --- the factor $N^\rho$ split off so that
	every summand lies in $[0,1]$ --- and $s=f(1)/f(\half)$ is computed exactly
	in $\widetilde O(R/h+B_{\mathrm{in}})$ time, where $B_{\mathrm{in}}=\sum_ib_i$
	is the input size; for extensive range $R=O(m)$ this is $\widetilde O(m/h)$.
	Moreover, at every level,
	\begin{equation}\label{eq:Abracket}
		\frac1{1+\rho}\ \le\ \frac{\Ac_\rho(t)}{\left(N\cdot\Fs(t)\right)^\rho}\ \le\ 1,
	\end{equation}
	with equality on the right iff the level is a singleton $(n(t)=1)$: the
	block-maximal surrogate $(N\cdot\Fs(t))^\rho$ is exact precisely when no two keys
	share a surprisal.
\end{theorem}

\begin{proof}
	The display is the level sum of \Cref{lem:rankcdf} (iii) under the scale
	split $r^\rho=N^\rho(r/N)^\rho$; the bracket \eqref{eq:Abracket} and the
	certified $O(1)$-per-level evaluation of $\Ac_\rho$ from
	$(\Fs(t^-),\Fs(t))$ are proved in \Cref{app:stable}. It remains to produce
	$\nu$ and to cumulate it into $\Fs$, with $\pi$ following from \eqref{eq:pinu}. A grid law is a coefficient vector,
	convolving two of them is polynomial multiplication, and the FFT performs it
	in $O(n\log n)$ rather than $O(n^2)$ on supports of $n$ points. Folding the
	$m$ per-coordinate laws in one at a time would re-traverse the growing
	support $\Theta(m)$ times, so instead we convolve them along a \emph{balanced
		binary merge tree}, as in mergesort --- coordinates in pairs, then the
	results in pairs, for $\lceil\log_2m\rceil$ rounds --- in exact arithmetic
	here (the floating-point realization is \Cref{sec:stable}), sizing each FFT
	to the actual support of its output. Supports of disjoint blocks add, so the
	nodes of any one round jointly hold $O(R/h+m)$ grid points and cost
	$O((R/h+m)\log(R/h+m))$; over the rounds this is
	$\widetilde O(R/h+m)$ as each grid cell is touched at $O(\log m)$ nodes.
	The cumulative sum and the two moment sums are linear.
\end{proof}

In~\Cref{fig:example}, the block-maximal surrogate would give $s\approx2.27$ instead of $2.30$.

\paragraph{The finite-rank exact engine.} Commensurability is the case $d=1$ of
a broader hypothesis under which the level count stays polynomial.
Say~the marginals have \emph{rank} $d$ if there are reals $\beta_1,\dots,\beta_d$,
linearly independent over $\mathbb Q$, constants $c_i$ and integers $c_{i,a,j}$ with
$s_i(a)=c_i+\sum_{j=1}^dc_{i,a,j}\beta_j$ for every $i,a$ --- the $c_i$ shift every
key's surprisal by the same $\sum_ic_i$, changing neither order nor moments. The rank and the
$\beta_j$ belong to the \emph{model specification} and are not recovered from
the numbers --- detecting $\mathbb Q$-linear relations among given reals is a
different and much harder problem --- but under structured advice they may come for
free: for a repeated $k$-symbol law the logarithms of its $k$ probabilities
span a $\mathbb Q$-vector space of dimension $d\le k$. A basis of that span,
with denominators cleared, puts the surprisals in the form above. Bounded $d$
is what divides the paper's two routes: ternary secrets sit at small $d$, while
heterogeneous leakage posteriors generically have rank $\Theta(m)$, so the latter regime belongs to the binned engine (\Cref{sec:binning}).

Bounded rank keeps the level count polynomial. Collect the integer vectors
$\mathcal W\coloneq \{\sum_ic_{i,x_i}:x\in K\}\subseteq\Zb^d$. The map
$w\mapsto\langle w,\beta\rangle$ is injective on $\Zb^d$ --- a coincidence of
values would be a nontrivial rational relation among the $\beta_j$ --- so the
levels of $S$ correspond exactly to the elements of $\mathcal W$, and these lie in a box
whose $j$-th side is the total spread of $c_{\cdot,\cdot,j}$ over the
coordinates. When those spreads are extensive, the rank-$d$ analogue of an
extensive range, the box holds $O(m^d)$ points, reducing to $O(R/h)$ at $d=1$
(\Cref{app:stable}).
The engine then carries over with $\Zb^d$ in place of $h\Zb$: per-coordinate
laws become measures on $\Zb^d$ supported on $b_i$ points, and the balanced
merge uses $d$-dimensional FFTs sized to the actual support boxes,
$\widetilde O(m^d)$ in total.
\Cref{alg:exact} below is stated for $d=1$ for simplicity.
For higher ranks, read $\Zb^d$ for $h\Zb$ throughout it.
Additionally, one step is genuinely new. At
$d=1$ the grid $h\Zb$ arrives already ordered, so cumulating $\nu$ into $\Fs$
is immediate; for $d\ge2$ the surprisal order is induced by the linear
functional $\langle\cdot,\beta\rangle$, and the $O(m^d)$ levels must be sorted
by $\langle w,\beta\rangle$ before they can be cumulated.
The finish is untouched: cumulate, recover $\pi$ from \eqref{eq:pinu}, apply \Cref{thm:exact}.

That sort is where exactness meets its limit. Ordering requires deciding the
sign of $\langle w-w',\beta\rangle$, a nonzero real, which adaptive-precision
interval arithmetic settles in finitely many refinements whenever the $\beta_j$
are available to arbitrary precision, as logarithms of rationals are --- but the
precision required is not bounded a priori. The fallback is a device this paper already
carries: levels that working precision cannot separate are merged into one block
and averaged by $\Ac_\rho$. Merging replaces the favorable pairing of largest
masses with smallest ranks by a neutral one, so by \Cref{lem:rankcdf}(i) it can
only \emph{overstate} the moment; the merged value is a certified upper bound.
The matching lower bound is elementary: a merged block occupies a
\emph{consecutive} rank interval $\{a+1,\dots,a+n\}$ of mass $\Pi$, so its true
contribution is at least $\Pi(a+1)^\rho$ against a merged value of at most
$\Pi(a+n)^\rho$. Summing these widths over the merged blocks, the engine returns
a certified enclosure in place of an exact value --- degrading into the same
two-sided mode as the binned route, by the same machinery. Where the level count
stays polynomial this engine is cheap, and it is precisely there that the binned
certificate's precision later floors (\Cref{rmk:certified}): the two are complementary
rather than competing, with a band at slowly growing $d$ in which both run and
neither dominates.

\subsection{Stable evaluation}\label{sec:stable}
At typical $m$ the dominant summands $\pi(t)\Fs(t)^\rho$ are $e^{-\Theta(m)}$
and underflow in double precision. But their smallness is benign: at every
level it factors into an exponential scale that is available in closed form,
times a shape of moderate dynamic range (see \eqref{eq:centring} below). The scheme below maintains this
factorization through its whole computation: floating point only ever
touches the shape, while the scale travels symbolically in the log domain; in
effect, a mantissa--exponent split for the entire pipeline. We state it first
and then justify its steps in turn. Three lemmas make it rigorous,
of which two are proved in \Cref{app:stable}.

\begin{algobox}{Stable exact evaluation (commensurate $P$)}\label{alg:exact}
	\emph{Input:} commensurate marginals on $h\Zb$.\\
	\emph{Output:} speedup $s$, accurate
	to floating point (\Cref{lem:stable}).
	\begin{enumerate}
		\item Form~the per-coordinate count laws $\nu_i$ on $h\Zb$.
		\item \textbf{For} each $\rho\in\{1,\half\}$ \textbf{do}
		      \begin{enumerate}
			      \item Tilt~every $\nu_i$ by $\phi_\rho\coloneq-1/(1+\rho)$, which centers the merged
			            law on the levels dominating $f(\rho)$ (see Eq.~\eqref{eq:centring}); by
			            \Cref{lem:tiltconv} this is exact.
			      \item FFT-merge the tilted $\nu_i^{(\phi_\rho)}$ up a balanced tree
			            (support-aware, truncated; \Cref{lem:fftcost}).
			      \item In the log domain, recover $\ln\Fs$ by the left-to-right
			            prefix $\LSE$ of Eq.~\eqref{eq:logrecover} and $\ln\pi$
			            from Eq.~\eqref{eq:pinu}, so the $e^{-\Theta(m)}$
			            magnitudes never leave the analytic factor.
			      \item Assemble
			      			$f(\rho)=\LSE_t[\ln\pi(t)+\ln(N^{-\rho}\Ac_\rho(t))]+\rho\ln N$,
			            each $N^{-\rho}\Ac_\rho(t)$ evaluated in $O(1)$ from the
			            prefix pair $(\Fs(t^-),\Fs(t))$ to certified precision
			            (shown in \Cref{app:stable}).
		      \end{enumerate}
		\item Output $s=f(1)/f(\half)$.
	\end{enumerate}
\end{algobox}

\paragraph{The tilt (step~\emph{a}).} The device effecting the split is the
exponential \emph{tilt}: for a grid law $\mu$ with cumulant generating function
(CGF) $\Lambda_\mu(\phi)=\ln\sum_we^{\phi w}\mu(w)$, set
$\mu^{(\phi)}(w)\coloneq e^{\phi w-\Lambda_\mu(\phi)}\mu(w)$, an exponential ramp
across the levels, re-normalized into a probability law with mean
$\Lambda_\mu'(\phi)$.
Choosing $\phi$ so that this mean lands on the levels dominating the moment sums (the \emph{saddle} levels) makes the ramp cancel the $e^{-\Theta(m)}$ exactly there: the entries the answer depends on become well-scaled entries of $\mu^{(\phi)}$, and the factor removed, $e^{\Lambda_\mu(\phi)-\phi w}$, is known exactly.
The~engine tilts the count law $\nu$, whose CGF is the uniform CGF $\Lam(\phi)=\ln\E_U[e^{\phi S}]=\sum_i\lambda_i(\phi)$ with per-coordinate summands $\lambda_i(\phi)=\ln\tfrac1{b_i}\sum_ae^{\phi s_i(a)}$.

\begin{lemma}[tilting]\label[lemma]{lem:tiltconv}
	\emph{(i)} Tilting commutes with convolution and CGFs add: $(*_i\mu_i)^{(\phi)}=*_i\mu_i^{(\phi)}$ and $\Lambda_{*_i\mu_i}=\sum_i\Lambda_{\mu_i}$.
	\emph{(ii)} Tilts compose additively, $(\mu^{(\phi)})^{(\phi')}=\mu^{(\phi+\phi')}$.
	\emph{(iii)} $\nu^{(\phi)}$ is the law of $S$ under the product law $U_\phi(x)\propto e^{\phi S(x)}=P(x)^{-\phi}$ (the law $U_\tau$ of \Cref{sec:accuracy}).
	As~$U_{-1}=P$, in particular $\nu^{(-1)}=\pi$, which is \eqref{eq:pinu} in tilt form, and $\pi^{(\phi+1)}=\nu^{(\phi)}$.
\end{lemma}
\begin{proof}
	\emph{(i)} Since $e^{\phi(a+b)}=e^{\phi a}e^{\phi b}$, multiplying $\mu_1*\mu_2$ by $e^{\phi w}$ convolves $e^{\phi a}\mu_1(a)$ with $e^{\phi b}\mu_2(b)$; summing over $w$ shows that the CGFs add, and normalizing gives the claim for two factors; induct.
	\emph{(ii)} Use~$\Lambda_{\mu^{(\phi)}}(\phi')=\Lambda_\mu(\phi+\phi')-\Lambda_\mu(\phi)$.
	\emph{(iii)} $U_\phi(x)=e^{\phi S(x)-\Lam(\phi)}/N$; sum over $\{S=t\}$, then use (ii).
\end{proof}

Convolving the tilted per-coordinate laws therefore returns \emph{exactly} the
tilt of the true law: tilting carries no bias, and solving its definition
recovers the true law, level by level.
We~call the per-coordinate laws of one kind, with their tilted partial convolutions up the merge tree, a \emph{chain}: the \emph{count chain} ($\nu_i\to\nu$) is cumulated into $\Fs$, while the \emph{mass chain} ($\pi_i\to\pi$) enters pointwise and only the binned engine (\Cref{sec:binning}) runs it.

\paragraph{The working tilt.} For~$f(\rho)$, step~\emph{a} tilts by $\phi_\rho=-1/(1+\rho)$, which makes the scale of every summand level-independent.
\begin{corollary}[Arıkan's scale, exactly]\label[corollary]{cor:arikan}
	For~$\rho>0$,
	\begin{equation}\label{eq:centring}
		\pi(t)\Fs(t)^\rho=N^{-\rho}e^{\rho H_{1/(1+\rho)}(P)}\,\nu^{(\phi_\rho)}(t)\,\mathcal M_\rho(t)^\rho
	\end{equation}
	at every level $t$, where $\mathcal M_\rho(t)\coloneq\sum_{u\le t}e^{-(t-u)/(1+\rho)}\,\nu^{(\phi_\rho)}(u)\le1$.
	Hence $f(\rho)=\rho H_{1/(1+\rho)}(P)+c_\rho$, where $c_\rho\coloneq\ln\E_{\nu^{(\phi_\rho)}}[\kappa_\rho\,\mathcal M_\rho^\rho]\in[-\rho\ln(1+\ln N),0]$ and $\kappa_\rho\coloneq\Ac_\rho/(N\cdot\Fs)^\rho\le1$ by \eqref{eq:Abracket}.
\end{corollary}
\noindent The~upper end of $c_\rho$ reproves Arıkan's upper bound and the lower end is his lower bound (proofs in \Cref{app:cscale}), so the engine evaluates exactly the slack left open by \eqref{eq:arikanscale}.
In~\Cref{fig:example}, $c_1\approx-1.26$ and $c_{1/2}\approx-0.69$, against lower ends $-1.56$ and $-0.78$, so the R\'enyi terms alone would give $2H_{1/2}(P)/H_{2/3}(P)\approx2.10$ instead of $s\approx2.30$.
The~centering is a numerical device, not a correctness ingredient: any $\phi$ keeping the contributing entries within the truncation window of \Cref{lem:fftcost} is admissible.

\paragraph{The merge (step~\emph{b}).} We~call a chain \emph{$\half$-decaying} if
within every coordinate the tilted atom weights fall off exponentially in the
surprisal, at a rate of at least $\half$ per unit above the per-coordinate
minimum.
The~count chain is $\half$-decaying at its working tilts $\phi_\rho\in[-\tfrac23,-\half]$, and so is the mass chain at $\phi_\rho+1$, whose atom weights are the same by \Cref{lem:tiltconv}(iii), after binning up to factors in $(1,e^{\eta}]$ (\Cref{app:stable}).
This~holds \emph{uniformly in the surprisal range}: the negative effective tilt is what makes long tails cheap.

\begin{lemma}[merge cost]\label[lemma]{lem:fftcost}
	\emph{(i)} The support-aware balanced merge of the $m$ per-coordinate laws,
	plus a cumulative sum, costs $\widetilde O(R/\eta+B_{\mathrm{in}})$ time and
	$O(R/\eta+m)$ space. \emph{(ii)} For $\half$-decaying chains, each
	partial convolution over a block $I$ of coordinates carries all but a
	$\poly^{-1}$ fraction of its mass on $\widetilde O(\sqrt{B_I}/\eta)$ grid
	points around its analytically known mean, where $B_I\coloneq \sum_{i\in I}b_i$ is
	the block input size; truncating every merge node to this window perturbs
	the output by $\poly^{-1}$ in $\ell_1$ (absorbed into \Cref{lem:stable}; see
	\Cref{app:stable}) and
	the cost drops to $\widetilde O\left((\sqrt{mB_{\mathrm{in}}}+m)/\eta\right)$
	time ($\widetilde O(m/\eta)$ for bounded alphabets) and
	$\widetilde O(\sqrt{B_{\mathrm{in}}}/\eta)$ space, independently of $R$.
\end{lemma}

\paragraph{Log-domain recovery (step~\emph{c}).} Inverting the tilt is the first line
of \eqref{eq:logrecover} below, and its two ingredients are precisely the split
announced above: an \emph{analytic factor} $-\phi w+\Lambda_\mu(\phi)$, of size
$\Theta(m)$ but known in closed form (\Cref{lem:tiltconv} assembles
$\Lambda_\mu=\sum_i\lambda_i$ from the marginals), and a computable
$\ln\mu^{(\phi)}(w)$, of moderate size by the centering. Applying this to $\nu$
and summing over the levels $w\le t$ gives the second line. We do this in the
log domain and hence use \enquote{log-sum-exps} ($\LSE$):
\begin{equation}\label{eq:logrecover}
	\begin{gathered}
		\ln\mu(w)=\underbrace{-\phi w+\Lambda_\mu(\phi)}_{\text{analytic, }\Theta(m)}
		+\ln\mu^{(\phi)}(w),\\
		\ln\Fs(t)=\Lam(\phi)+\LSE_{w\le t}\left[-\phi w+\ln\nu^{(\phi)}(w)\right],
	\end{gathered}
\end{equation}
where $\LSE_w[x_w]\coloneq \ln\sum_we^{x_w}$.
Log-sum-exps appear at two points, and are evaluated differently because they
are needed differently. The assembly of $f(\rho)$ in step~\emph{d} requires a
\emph{single} $\LSE$ over all levels, computed once. There the batch form
applies: shift by the maximum,
$\LSE_w[x_w]=\max_wx_w+\ln\sum_we^{x_w-\max_ux_u}$, so that every exponent is
$\le0$. However, the cumulation into $\ln\Fs$ in this step is another problem: it is
needed at \emph{every} level $t$, and treating each as a separate batch $\LSE$
would cost quadratic time in the number of levels. Instead a single
left-to-right sweep maintains the running prefix $\mathcal L_t\coloneq \LSE_{w\le t}[x_w]$,
folding in one level at a time by the identity
\[
	\mathcal L_t=\max\left(\mathcal L_{t-\eta},x_t\right)
	+\ln\left(1+e^{-\left|\mathcal L_{t-\eta}-x_t\right|}\right).
\]
The larger of the two numbers is
carried outside the logarithm unchanged; the smaller enters only through the
correction term, which lies in $(0,\ln2]$. The exponent is never positive, so
nothing overflows and each level costs $O(1)$.

\paragraph{Assembly and stability (step~\emph{d}).} At the assembly, we have:
\begin{lemma}[floating-point stability]\label[lemma]{lem:stable}
	\Cref{alg:exact} returns $\widehat s$ with
	\[
	|\widehat s-s_\eta|
	=
	O\left(
	\frac{\umach\,\poly(m,1/\eta)\,(1+s_\eta)}
	     {f_\eta(\half)}
	\right),
	\]
	where $\umach$ is the floating-point unit roundoff; $s_\eta$ and
	$f_\eta(\half)$ are the exact values of the estimator being run
	($s_\eta=s$ and $f_\eta=f$ in the commensurate case), provided the
	roundoff in $f_\eta(\half)$ is smaller than a constant fraction of
	$f_\eta(\half)$. For families with $f_\eta(\half)=\Theta(m)$ and
	$s_\eta=O(1)$, this simplifies to
	$|\widehat s-s_\eta|=O(\umach\,\poly(m,1/\eta)/m)$.
\end{lemma}

This~is negligible against target precisions $\delta\gg\umach\cdot\poly/m$; in IEEE double, $\umach\approx 1.1\times10^{-16}$.
The polynomial in \Cref{lem:stable} is the only constant in the pipeline whose
growth in $m$ and $1/\eta$ is unspecified; those of \Cref{sec:accuracy} are absolute.
It governs floating-point roundoff, not the binning error that
\Cref{prop:cert} later certifies, and it is removable at the implementation level:
the entire assembly is a composition of FFTs, sums and $1$-Lipschitz $\LSE$s,
so running it in directed-rounding interval arithmetic returns a rigorous
enclosure end-to-end, and the certificate of \Cref{prop:cert} then certifies its own arithmetic.
The~time overhead is a constant factor at fixed working precision, but the widths of the resulting enclosures are not analyzed here, so we include this information as a possible implementation route only.

The full proofs of both lemmas are deferred to \Cref{app:stable}, so here we focus
on delineating their underlying mechanism: after tilting, the grid entries
the output depends on sit within a $\poly$ factor of the norms of the vectors
the FFTs handle. Thus, \citeauthor{higham_accuracy_2002}'s normwise bound, which is $O(\umach\log L)$ for
a length-$L$ transform~\cite{higham_accuracy_2002}, accumulated over the
$m-1$ internal nodes of the merge tree, leaves them with $\umach\cdot\poly$ relative
precision. The $e^{-\Theta(m)}$ saddle magnitudes never materialize as floats,
living instead in the analytic factor of \eqref{eq:logrecover}, so no
contributing value underflows; the log-domain assembly is a composition of $1$-Lipschitz $\LSE$s;
and the final ratio contributes the conditioning factor $(1+s_\eta)/f_\eta(\half)$.
\section{Binning: the general case}\label{sec:binning}
A measured posterior is generically \emph{not} commensurate: each coordinate has
only $b_i$ surprisal values, but the total $S$ takes up to $\prod_i b_i$
distinct values in a bounded range, so direct enumeration is infeasible. The
reduction is essentially one step: snap the per-coordinate surprisals to a common grid
$\eta\Zb$, forcing commensurability, and run the engine (see \Cref{alg:binned}).

\paragraph{Bin the surprisals, not the probabilities.} Round each per-coordinate
surprisal \emph{strictly} up onto $\eta\Zb$,
$\widetilde s_i(a)\coloneq\eta(\lfloor s_i(a)/\eta\rfloor+1)$, which lies in $(s_i(a),s_i(a)+\eta]$,
so that $0<\tS(x)-S(x)\le m\eta$ for every key.%
\footnote{Strictness --- surprisals already on the grid are lifted by a full
	$\eta$ rather than fixed --- is what places the optimal rank in the window
	of \Cref{lem:sandwich} at no cost to the constants.
	Round-to-nearest remains admissible: its offsets are two-sided, $|\tS-S|\le m\eta/2$, and re-running the proof of \Cref{lem:sandwich} with the two-sided bound widens the window to $(m+1)\eta$, so the same $C(s)$ holds with $\eps_m$ increased by $1/m$.
	Under either rounding mode, the actual error $|s_\eta-s|$ that we observed against exhaustive enumeration sat orders of magnitude below the certified bound and shrank faster than linearly in $\eta$.
	We~do not claim that these observations hold in general, and nothing in this paper relies on them.} It is
essential to bin the surprisals rather than the probabilities: rounding
probabilities would place the surprisals on an irregular set
whose $m$-fold sums have super-polynomially many
distinct values and do not collapse. Snapping surprisals to $\eta\Zb$ puts
every sum on $\eta\Zb$, with support $O(R/\eta)$.

\paragraph{Two chains and the estimator.} Binning breaks Eq.~\eqref{eq:pinu}: the binned mass law $\tpi_i(k\eta)\coloneq\sum_{a:\widetilde s_i(a)=k\eta}p_{i,a}$ keeps the exact masses, so it matches the tilt $\tnu_i^{(-1)}$ of the binned count law $\tnu_i(k\eta)\coloneq\#\{a:\widetilde s_i(a)=k\eta\}/b_i$ only up to factors in $(e^{-\eta},e^{\eta})$, compounding to $e^{\pm m\eta}$ over the coordinates.
The~binned engine therefore runs both chains, yielding $\tnu$ (law of $\tS$ under $U$, cumulated into $\tF$) and $\tpi$ (law of $\tS$ under the true weights $P$).
The~estimator is the finite sum of \Cref{thm:exact} with $(\tpi,\tF)$ in place of $(\pi,\Fs)$: with
$\tAc_\rho(t)$ the level mean power rank of the binned law (built from
$\tF(t^-),\tF(t)$ exactly as in \Cref{lem:rankcdf}),
\begin{equation}\label{eq:estimator}
	\begin{gathered}
		f_\eta(\rho)\coloneq \rho\ln N+\ln\sum_t\tpi(t)\,\frac{\tAc_\rho(t)}{N^\rho}
		=\ln\E_P\left[\tAc_\rho(\tS)\right],\\
		\tG(x)\coloneq N\cdot\tF(\tS(x)),\qquad s_\eta\coloneq f_\eta(1)/f_\eta(\half),
	\end{gathered}
\end{equation}
where $\tG$, the binned block-maximal rank, recurs in the analysis below.
The outer measure stays the true $P$; only the surprisal \emph{values} entering
the rank are binned. This is what makes the error analysis of
\Cref{sec:accuracy} rank-only.

\begin{algobox}{Certified binned evaluation (general $P$)}\label{alg:binned}
	\emph{Input:} marginals $\{p_i\}$, target $\delta>0$.\\
	\emph{Output:}
	estimate $\widehat s$ and a certified bound $B\le\delta$ on $|s_\eta-s|$.
	\begin{enumerate}
		\item Choose a bin width $\eta_0\ll\min_i\operatorname{gap}_i$, the smallest
		      nonzero per-coordinate surprisal gap.
		\item Round each $s_i(a)$ strictly up to $\eta_0\Zb$ and form the binned laws $\tnu_i,\tpi_i$.
		\item Recover $\ln\tF$ from the binned count chain exactly as in steps (a)--(c) of \Cref{alg:exact}.
		      Compute $\ln\tpi$ separately from the binned mass chain in a similar fashion, tilted by $\phi_\rho+1$ and not cumulated, since binning breaks Eq.~\eqref{eq:pinu}.
		      Assemble $f_\eta(1)$ and $f_\eta(\half)$ from $(\tpi,\tF)$ as in step~(d), which is Eq.~\eqref{eq:estimator}.
		\item Assemble the certificate $B$ on $|s_\eta-s|$ of \Cref{prop:cert} from $(\tpi,\tF,\tAc_\rho)$, in one pass over arrays already in memory.
		\item If $B>\delta$, rerun steps 2--4 at $\eta_1=\eta_0\,\delta/B$,
		      halving instead whenever that prediction under-delivers.
		\item Report $\widehat s$, $B$ as well as the diagnostics $\Lameta$
		      of Eq.~\eqref{eq:lameta} and the empty-lattice fraction near the mean.
	\end{enumerate}
\end{algobox}

Three of these steps carry a caveat worth stating plainly. The bin width in
step~1 must resolve the per-coordinate structure: coarser than the gaps,
binning collapses distinct levels and the certificate, while still valid,
becomes vacuous. No a priori constant enters that choice --- only the
marginals, which are given. The prediction in step~5 is motivated by the
scaling law's linearity in $\eta$, but nothing rests on it: an over-optimistic
$\eta_1$ simply returns a $B$ still above $\delta$ and the schedule continues,
so correctness never depends on the prediction (only termination speed does).
And the diagnostics of step~6 must be recomputed at \emph{this} $\eta$: the
band $|t|\le\pi/\eta$ widens as $\eta$ shrinks, so a $\Lameta$ computed at a
coarser bin width does not transfer. The empty-lattice fraction is a rigorous
falsifier: an empty grid point near the mean will contradict \Cref{lem:lclt}
whenever $\eps_\tau<0.399$.

\begin{theorem}[precision and cost]\label{thm:precision}
    \emph{(i) Guarantee.} For every bin width $\eta$, the exact-arithmetic estimator satisfies $|s_\eta-s|\le B$ whenever $f_\eta(\half)>E_{1/2}$, where $B$ is the computable certificate \eqref{eq:certbound} of \Cref{prop:cert}. The floating-point realization incurs in addition the error of \Cref{lem:stable}; under directed-rounding interval arithmetic the resulting enclosure is rigorous end-to-end.
    \emph{(ii) Scaling and cost.} Under the
	scaling-law hypotheses of \Cref{thm:bin}, the estimator satisfies the
	\emph{multiplicative} scaling law
	\begin{equation}\label{eq:Cbound}
		\boxed{\;|s_\eta-s|\ \le\ C(s)\,\eta\,(1+\eps_m),\qquad
			C(s)=\frac{2+s}{\Hbar_{2/3}},\;}
	\end{equation}
    with $\eps_m$ computed as in \Cref{thm:bin}, and $\eps_m=O(m^{-1/2})$ under the regularity conditions of \Cref{cor:scaling}; the returned bound
    tracks it, $B=O(C(s)\,\eta)$; with truncated merges, this holds when $m\eta=O(\log(m/\eta))$ and, under these conditions, throughout the regime (\Cref{app:cert}).
	Consequently, for uniformly bounded alphabets, the refinement schedule of \Cref{alg:binned} reaches every
	target $\delta$ above the model's atom-spacing floor
	(\Cref{rmk:certified}) at final bin width
	$\eta_{\mathrm{final}}=\Theta(\delta/C(s))$ and total cost
	$\widetilde O(m/\eta_{\mathrm{final}})
		=\widetilde O(m\,C(s)/\delta)$, the geometric schedule being
	dominated by its final run; the cost is independent of the surprisal
	range.
\end{theorem}

The leading constant is $\Theta(1)$ in $m$ and free of the range $R$ for any
\emph{fixed} per-coordinate law, but it is not bounded over families: it grows
as $\Hbar_{2/3}\to0$, and high-confidence advice --- where a single symbol
carries almost all the mass, placing large surprisals in the tail --- can
inflate it by more than an order of magnitude over a mildly biased coordinate.
The cost $\widetilde O(mC(s)/\delta)$ inherits that growth. The certificate
route is indifferent to it: it reads its bound off the run, and a large $C(s)$
appears there as a larger returned $B$ at the same $\eta$. The theorem is
proved in the next section (both scaling law and certificate) together with
\Cref{lem:stable} (floating-point error).

\section{Accuracy of the binned estimator}\label{sec:accuracy}
We bound $|s_\eta-s|$ and prove \Cref{thm:precision}. The strategy: an exact,
CLT-free \emph{sandwich} traps the true and the binned rank of every key in a
\emph{common} window of the binned (lattice) CDF $\tF$. The window is anchored
at the true surprisal and extends (only) upward by the worst-case rounding shift
$m\eta$; nevertheless the resulting error is sign-free since the relative order of the
two ranks inside the window is uncontrolled. The sandwich thus reduces the
binning error to a log-ratio of $\tF$ across that window (\Cref{sec:sandwich}).
We~then use exponential tilting at the level saddle to identify the local growth rate of that lattice CDF (\Cref{sec:rate}).
In~\Cref{sec:worst}, a value-level estimate of the discrete reversed hazard caps this rate over $\{t\ge H\}$ through the exact $|\tau_t|\le1$ of \Cref{cor:rate}.
Together these steps yield the two-sided worst-case bound \eqref{eq:Cbound}, a scaling law with computed constants.
\Cref{sec:cert} adds a non-asymptotic certificate that each run assembles from its own arrays, using none of the scaling law's hypotheses.

\subsection{The lattice shift sandwich}\label{sec:sandwich}
Herein, the key step is to compare both ranks against the \emph{binned} CDF $\tF$
(a genuine lattice law) so that the analysis never touches the true,
generically non-lattice $\Fs$ except through one elementary, CLT-free
inequality. A sandwich is genuinely needed because binning is \emph{not}
order-preserving: with $\eta=1$, per-coordinate surprisals $(1.0,\,0.7)$ sum
to $1.7$ and round strictly up to $3$, while $(0.9,\,0.9)$ sum to $1.8$ and
round to $2$ --- bins reorder keys, so no exactness argument is available in
the general case and the window below does real work.

\begin{lemma}[lattice shift sandwich]\label[lemma]{lem:sandwich}
	For every key $x$, the true optimal rank interval
	$J(x)=(N\cdot\Fs(S(x)^-),\,N\cdot\Fs(S(x))]$ and the binned block
	$\widetilde J(x)=(N\cdot\tF(\tS(x)^-),\,N\cdot\tF(\tS(x))]$ both lie in the
	common binned window $W(x)\coloneq (N\cdot\tF(S(x)),\ N\cdot\tF(S(x)+m\eta)]$.
	Hence, for every $\rho>0$, we have $|\ln\tAc_\rho(\tS(x))-\rho\ln G(x)|\le\rho\,\Delta(x)$
	with $\Delta(x)\coloneq \ln\frac{N\cdot\tF(S(x)+m\eta)}{\max(N\cdot\tF(S(x)),\,1)}$,
	a log-ratio of the lattice $\tF$ floored at rank $1$.
\end{lemma}
\noindent The~proof (\Cref{app:sandwich}) rests on the CLT-free comparison $\tF(t)\le\Fs(t^-)\le\Fs(t)\le\tF(t+m\eta)$ for all $t$, which strict rounding provides.
The size of $\Delta(x)$ is set by how fast $\tF$ grows across a window of width
$m\eta$ at the level $t=S(x)$, i.e.\@ the \emph{rank-growth rate}.

\begin{remark}[sublattices, resonances, and the computed test]\label[remark]{rmk:sublattice}
	If all binned surprisals lie on a common sublattice of span $k\eta$ (a
	genuine \emph{coarse sublattice}), nothing degrades: the leading
	$C(s)\eta$ is set by the rounding resolution regardless of where the binned
	values land, and the local CLT applies at span $k\eta$ instead --- one
	takes the gcd of the binned surprisals to identify the true span and
	computes $\Lameta$ on the matching band. What the gcd test cannot see is
	\emph{nearly}-commensurate advice, and this is neither undetectable nor an
	edge case: it is precisely the event $\Lameta\ll m$, and it is detected by
	computing $\Lameta$. A concrete instance with bounded leakage and
	$\gcd=1$: at $\eta=0.01$ and $m=64$, take $57$ coordinates Bernoulli$(0.8)$ and
	$7$ Bernoulli$(0.5025)$ --- binned gaps of $138$ and $1$ lattice steps,
	so the summed lattice genuinely has span $\eta$; yet $\Lameta=0.0036$ at $\tau=0$,
	and $94\%$ of the bulk lattice points carry \emph{exactly} zero
	mass, so the span-$\eta$ local CLT's conclusion is false there. Neither
	$\gcd=1$ nor bounded per-coordinate leakage implies the hypothesis; only
	the computed $\Lameta$ decides it. At the opposite extreme, advice on a
	\emph{proper} coarse sublattice, $h=k\eta$ with $k\ge2$, forces $\Lameta=0$ on the span-$\eta$ band:
	the resonance $t=2\pi/h$ then lies inside the band $|t|\le\pi/\eta$, where
	every $\sin^2(t\Delta_{iab}/2)$ vanishes --- so the scaling law's
	hypothesis fails where \Cref{thm:exact} applies. This is no
	loss: those inputs belong to the exact engine (\Cref{sec:exact}), not the
	binned route. In every case the certificate of \Cref{prop:cert} remains
	valid; only the scaling law reads $\Lameta$.
\end{remark}

\subsection{The level-saddle apparatus and the rate}\label{sec:rate}
\Cref{lem:sandwich} reduced the binning error to a single quantity: the growth
$\ln[\tF(t+m\eta)/\tF(t)]$ of the lattice CDF across the rounding window. Both
entries are lower-tail probabilities of size $e^{-\Theta(m)}$, so we compute
the ratio by exponential tilting, i.e.\@ reweighting $U$ after which the rare level
$t$ becomes typical, so that the ratio is a bulk quantity governed by a
central limit theorem rather than a tail estimate.
For~$t<\mu_U$ the \emph{saddle}
$\tau_t<0$ is the unique root of $\Lam'(\tau_t)=t$: the tilted law
$U_{\tau_t}(x)\propto e^{\tau_tS(x)}$ has mean surprisal exactly $t$. As $t$
increases, $\tau_t$ increases with it, which means that $t\mapsto|\tau_t|$ is strictly decreasing.
A one-line change of measure ($dU/dU_{\tau_t}=e^{-\tau_tS+\Lam(\tau_t)}$)
makes the tail–bulk exchange exact: with $\tau=\tau_t$, $I_U(t)\coloneq\tau_tt-\Lam(\tau_t)$, and $\bar S=S-t$, centered under
$U_{\tau_t}$,
\begin{equation}\label{eq:tiltid}
	\Fs(t+a)=e^{-I_U(t)}\,M(a),\qquad
	M(a)\coloneq \E_{U_{\tau_t}}\!\left[e^{|\tau_t|\bar S};\,\bar S\le a\right].
\end{equation}
\eqref{eq:tiltid} factors the tail probability into the
large-deviation cost $e^{-I_U(t)}$, which carries the entire exponential
$t$-dependence, and a window term $M(a)$ whose weight $e^{|\tau_t|\bar S}$ is
the likelihood ratio left over after extracting the cost at the anchor $t$.
Since the tilted point masses are nearly flat near the mean while the weight
decays by $e^{-|\tau_t|\eta}$ per lattice step down, $M(a)$ is a
geometric-type sum dominated by the last $O(1/|\tau_t|)$ of surprisal below
$t+a$, which is the origin of the growth rate $|\tau_t|$ below. In the ratio
$\Fs(t+m\eta)/\Fs(t)=M(m\eta)/M(0)$ the cost cancels exactly, leaving a
quantity governed by the tilted law near its mean.
The~same identity holds verbatim for the binned chain, tilted by $\widetilde U_\tau$ at its own saddle $\widetilde\tau_t$, the root of $\tLam'(\widetilde\tau_t)=t$.
The~one analytic input is a \emph{lattice} local CLT for the tilted law at the density level: for $|\tau_t|=\Theta(1)$ and $\sigma_t^2=\Theta(m)$, $M(a)$ is only of order $e^{|\tau_t|a}/\sqrt m$, as is the error that a CDF-level Berry--Esseen bound would leave in it (\Cref{rmk:cdfbe}).

\begin{lemma}[tilted/lattice local CLT, explicit]\label[lemma]{lem:lclt}
	Under the binned tilt $\widetilde U_\tau$, $\widetilde U_\tau^{(i)}(a)\propto e^{\tau\widetilde s_i(a)}$, for $\tau\in[\widetilde\tau_H,0]$ (so that $\tS$ has law $\tnu^{(\tau)}$), let $\bar S=\sum_i(Y_i-\E Y_i)$ with
	$Y_i=\widetilde s_i(X_i)$ the binned per-coordinate surprisals: a centered
	lattice variable of span $\eta$ (on a grid offset by the centering), with
	variance $\tsig_\tau^2=\tLam''(\tau)$ and summed third moment
	$\beta_\tau=\sum_i\E_{\widetilde U_\tau}|Y_i-\E Y_i|^3$. Write $g_k=p_k/\eta$ with
	$p_k=\PP_{\widetilde U_\tau}[\bar S=s_k]$, let $\psi_i$ be the characteristic function
	of the centered $Y_i$, with variance $\sigma_i^2$, and set $\delta_1\coloneq\min\{3\tsig_\tau^2/(2\beta_\tau),\,\sqrt3/(2\max_i\sigma_i)\}$.
	Define the \emph{lattice-span exponent}
	\begin{equation}\label{eq:lameta}
		\Lameta\coloneq \min_{\delta_1\le|t|\le\pi/\eta}\
		\sum_i\sum_{a<b}4\,p^{(\tau)}_{i,a}\,p^{(\tau)}_{i,b}\,
		\sin^2\!\left(t\,\Delta_{iab}/2\right),
	\end{equation}
	(if $\delta_1\ge\pi/\eta$ the band is empty, $\Lameta\coloneq+\infty$, and the
	high-frequency contribution vanishes), with $\Delta_{iab}$ the binned surprisal gaps of coordinate $i$ and
	$p^{(\tau)}_{i,a}\coloneq\widetilde U_\tau^{(i)}(a)$ the tilted weights. Then
	$|\prod_i\psi_i(t)|\le e^{-\Lameta/2}$ on the band, and
    $\sup_k|g_k-\tfrac1{\tsig_\tau}\varphi(s_k/\tsig_\tau)|
	\le\frac{\eps_\tau}{\tsig_\tau}$
    with
    $\eps_\tau=\frac{4}{3\pi}\frac{\beta_\tau}{\tsig_\tau^3}
	+q_\tau+\frac{\tsig_\tau}{\eta}\,e^{-\Lameta/2}
	+\sqrt{\tfrac2\pi}\,\Gbar(\tsig_\tau\delta_1)$,
	where $q_\tau$ is the quartic correction (a computed integral over the
	marginals, see \Cref{app:lclt}). Every term of $\eps_\tau$ is a computed
    number; under the regularity conditions of \Cref{cor:scaling} the asymptotic
	reading is	$\eps_\tau=\frac4{3\pi}\beta_\tau/\tsig_\tau^3+O(1/m)+O(e^{-\Theta(m)})
		=\Theta(m^{-1/2})$ (\Cref{cor:scaling}).
	$\varphi$ and $\Phi$ denote the standard normal density and CDF, and $\Gbar=1-\Phi$.
\end{lemma}
\noindent The full proof is in \Cref{app:lclt}. The intuition: after tilting, the
histogram of $\bar S$ tracks the Gaussian curve pointwise. Each lattice
point carries mass $p_k\approx\eta\,\varphi(s_k/\tsig_\tau)/\tsig_\tau$
with an error uniform in $k$; hence the \emph{relative} error is $O(m^{-1/2})$
at points within $O(\tsig_\tau)$ of the mean, precisely the points
dominating the sums $M(a)$. Three usage notes. \emph{(a)} Only the scaling
law (\Cref{thm:bin}, \Cref{lem:br}) consumes this lemma; the
computable certificate (\Cref{prop:cert}) never touches it.
\emph{(b)} The~asymptotic reading is correct but takes hold only at large $m$: the high-frequency term $\tsig_\tau\eta^{-1}e^{-\Lameta/2}$ of $\eps_\tau$ stays above one until $\Lameta\ge2\ln(\tsig_\tau/\eta)$, so with $\Lameta\approx\lambda m$ the lemma has content only once $m$ is of order $\lambda^{-1}\ln(\tsig_\tau/\eta)$, a threshold that refining $\eta$ raises.
\emph{(c)} The~scalars are quantified over the tilt range $\tau\in[\widetilde\tau_H,0]$, which contains the binned saddles of all levels in $[H,\mu_U]$.
Its~edge $\widetilde\tau_H$ differs from $\tau_H=-1$ only by the computed saddle shift of \Cref{app:rate}.

The~error analysis needs the rank-growth rate $|\tau_t|$ only as an upper bound, one lattice step at a time, which \Cref{lem:br} supplies with an explicit relative error.
At~the edge level $t=H$ the rate takes an exact, advice-independent value.

\begin{corollary}[the exact rank-growth rate]\label[corollary]{cor:rate}
	For every non-degenerate product posterior (standing assumption~(iv),
	$\Var_U(S)>0$), $\theta^\star\coloneq |\tau_H|=1$, where
	$\Lam'(\tau_H)=H$. Hence $|\tau_t|\le\theta^\star=1$ for
	$H\le t\le\mu_U$; beyond $\mu_U$ the rate is covered by the
	outside-saddle estimate of \Cref{lem:br}.
\end{corollary}
\begin{proof}
	$U_\tau^{(i)}(a)\propto e^{\tau s_i(a)}=p_{i,a}^{-\tau}$ gives $U_{-1}=P$ and
	$\Lam'(-1)=\E_P[S]=H$ \emph{unconditionally}; non-degeneracy makes $\Lam$
	strictly convex, so the root of $\Lam'(\tau)=H$ is unique, whence $\tau_H=-1$
	and $\theta^\star=1$. Since $\Lam'$ is increasing with
	$\Lam'(-1)=H$ and $\Lam'(0)=\mu_U$, $|\tau_t|$ decreases from $1$ to $0$
	across $[H,\mu_U]$, which gives the bound there.
\end{proof}

This $\theta^\star=1$ is the asymptotic-equipartition statement~\cite{cover_elements_2006} that $\Gst\approx1/P=e^S$ for typical
keys ($d\ln\Gst/dS=1$ at $S=H$).

\subsection{The edge term and the worst-case bound}\label{sec:worst}
The~worst-case bound needs the rank-growth rate over all levels a contributing key can occupy.
Only~$\{t\ge H\}$ really matters (\Cref{thm:bin} shows the sub-entropy region carrying
exponentially little weight), and there the exact convexity inequality
$|\tau_t|\le|\tau_H|=1$ of \Cref{cor:rate} caps the rate outright: the worst
key for binning is the \emph{typical} one. The next lemma is a
\emph{value-level} estimate of the single-step discrete reversed hazard
$\tnu(t)/\tF(t)$.

\begin{lemma}[discrete reversed hazard, value-level]\label[lemma]{lem:br}
    Fix~$\tau_0\in(0,1)$.
    For~$t\ge H$ in the saddle range $|\tau_t|\ge\tau_0$, and under the computed
	small-error condition $\sqrt{2\pi}\,\eps_\tau\le c_0<1$, the binned law satisfies
	$\tnu(t)/\tF(t)\le|\tau_t|\,\eta\,(1+\eps_r(t))$,
    with $\eps_r(t)\le2\sqrt{2\pi}\,\eps_\tau(1+O(\eps_\tau))+O(z_t^{-2})+O(|\tau_t|\eta)+O(|\widetilde\tau_t-\tau_t|/|\tau_t|)$
    where $z_t\coloneq|\tau_t|\sigma_t$ and $\eps_\tau$ is from \Cref{lem:lclt}.
    Every scalar in $\eps_r$ is computed and the constants are absolute.
    On~the rest of $\{t\ge H\}$, that is, for $|\tau_t|<\tau_0$ and for $t\ge\mu_U$, the ratio is at most $(\tau_0+|\widetilde\tau_t-\tau_t|+O(1/\sigma_t))\,\eta\,(1+O(\eps_\tau))$, with $\tau_t$ and $\widetilde\tau_t$ read as $0$ where no negative saddle exists; this is below $\eta$ once the saddle shift, $1/\sigma_t$ and $\eps_\tau$ are small compared with $1-\tau_0$.
    Consequently, under this condition and with $|\tau_t|\le1$ on $[H,\mu_U]$ (\Cref{cor:rate}), we have
	$\sup_{t\ge H}\tnu(t)/\tF(t)\le\eta\,(1+\eps_r)$, with $\eps_r$ the supremum of $\eps_r(t)$ over the saddle range.
\end{lemma}

The~proof (\Cref{app:rate}) reads the single-step ratio directly off the tilting identity and \Cref{lem:lclt}'s point-mass control.

\begin{theorem}[binning error --- the scaling law]\label[theorem]{thm:bin}
Assume the scaling-law hypotheses: the regime $\sqrt m\,\eta\le\eps_0$ for a fixed small constant $\eps_0$; the lattice-span condition
$\Lameta\ge\lambda m$ on $[\widetilde\tau_H,0]$ (\Cref{lem:lclt}); the entropy gap
$\Hbar_{1/(1+\rho)}-\Hbar_1\ge g>0$ for $\rho\in[\half,1]$ (advice
non-uniform on a constant fraction of coordinates); $\log(1+\ln N)=o(m)$,
so that Arıkan's additive slack is lower-order; $\log(1/\eta)=o(m)$, so that the exponentially small contribution of the keys below the entropy level is negligible against $m\eta$; and the computed small-error
condition $\sup_{\tau\in[\widetilde\tau_H,0]}\sqrt{2\pi}\,\eps_\tau\le c_0<1$ of \Cref{lem:lclt}, together with the condition of \Cref{lem:br} off the saddle range.
Then, with $\theta^\star=1$
(\Cref{cor:rate}) and $f(\half)=\half m\Hbar_{2/3}(1+o(1))$ from
\eqref{eq:arikanscale}, we have
$|s_\eta-s|\ \le\ \frac{2+s}{\Hbar_{2/3}}\,\eta\,(1+\eps_m)$,
which is the scaling law behind \eqref{eq:Cbound}; $\eps_m$ reduces to computed
scalars with absolute constants, through \Cref{lem:lclt}, by the propagation
\begin{multline*}
	\eps_m\le2\sqrt{2\pi}\,\eps_\tau(1+O(\eps_\tau))+O(\bar z^{-2})+O(\bar\vartheta)\\
	+O(\eta)+O\left(\log(1+\ln N)/f_\eta(\half)\right),
\end{multline*}
where $\bar z^{-2}$ and $\bar\vartheta$ are the suprema of $z_t^{-2}$ and of the relative saddle shift $|\widetilde\tau_t-\tau_t|/|\tau_t|$ (\Cref{lem:br}) over the saddle range (\Cref{app:rate}).
Under the regularity conditions of \Cref{cor:scaling},
$\eps_m=O(m^{-1/2})$.
\end{theorem}
The constant is transparent: a key's log-rank moves by at most
$\theta^\star m\eta=m\eta$ (rate $\times$ window), so each moment exponent
$f(\rho)$ moves by at most $\rho m\eta$; the ratio $s=f(1)/f(\half)$ divides
by $f(\half)=\half m\Hbar_{2/3}(1+o(1))$, the $m$'s cancel, and what survives
is the per-coordinate resolution $\eta$ measured in units of the
per-coordinate entropy $\Hbar_{2/3}$. The factor $2+s$ merely collects the
numerator and denominator contributions ($e_1-s\,e_{1/2}$ in the proof).
The~proofs of this theorem and of the following corollary are in \Cref{app:rate}.

\begin{corollary}[interpretive scaling]\label[corollary]{cor:scaling}
	If the alphabets are uniformly bounded, the per-coordinate leakage is bounded,
	$\max_as_i(a)-\min_as_i(a)\le c$ for all $i$, and the range is extensive,
	$R=\Theta(m)$, then uniformly over the tilt range $\tau\in[\widetilde\tau_H,0]$:
	$\tsig_\tau^2=\Theta(m)$, $\beta_\tau\le(c+\eta)\,\tsig_\tau^2$, hence
	$\beta_\tau/\tsig_\tau^3=O(m^{-1/2})$ (Lyapunov); if moreover
	$\Lameta\ge\lambda m$ and $\log(1/\eta)=o(m)$, then $\eps_\tau=\Theta(m^{-1/2})$
	and $\eps_m=O(m^{-1/2})$.
\end{corollary}

\subsection{The computable certificate}\label{sec:cert}
The~scaling law bounds the error but does not let a reader holding the
marginals and $m$ \emph{evaluate} that bound.
Yet~every asymptotic ingredient above
is a functional of $\tF$ --- and $\tF$ is an output. The growth rate that
\Cref{lem:br} estimates can be \emph{measured} on the computed $\tF$; the
denominator that \eqref{eq:arikanscale} approximates is, up to $e_{1/2}$,
the computed $f_\eta(\half)$; and the localization threshold can be
dispensed with level by level. The result is an a-posteriori certificate in
the style of validated numerics: an interval, not a formula.

\begin{proposition}[computable certificate]\label[proposition]{prop:cert}
	Run the engine at bin width $\eta$ with strict upward rounding and let
	$(\tpi,\tF,\tAc_\rho)$ be its arrays.
	Under the truncated merge of \Cref{lem:fftcost}(ii) both chains discard mass.
	Both~losses are known exactly, and \Cref{app:cert} turns them into an enclosure of $\tF$, whose lower end enters $\mathrm{lo}_\rho$ and upper end $\mathrm{hi}_\rho$, and into one additive term on the upper end of \eqref{eq:bracket}, which then holds with the computed arrays.
	For each level $t$ of $\tS$ define
	\begin{equation}\label{eq:lohi}
		\begin{gathered}
			D(t)\coloneq \max_{u\in\{t-m\eta,\dots,t-\eta\}}
			\ln\frac{\tF(u+m\eta)}{\tF(u)}
			\quad(\infty\ \text{if some}\ \tF(u)=0),\\
			\mathrm{lo}_\rho(t)\coloneq \max\left\{\left(N\cdot\tF(t-m\eta)\right)^\rho,\ 1,\
			\tAc_\rho(t)\,e^{-\rho D(t)}\right\},\\
			\mathrm{hi}_\rho(t)\coloneq \min\left\{\left(N\cdot\tF(t+m\eta)\right)^\rho,\
			e^{\rho t},\ \tAc_\rho(t)\,e^{\rho D(t)}\right\},
		\end{gathered}
	\end{equation}
	dropping the $D$-candidates where $D(t)=\infty$. Then
	\begin{equation}\label{eq:bracket}
		\sum_t\tpi(t)\,\mathrm{lo}_\rho(t)\ \le\ e^{f(\rho)}\ \le\
		\sum_t\tpi(t)\,\mathrm{hi}_\rho(t),
	\end{equation}
	and, with $E_\rho$ the induced computed bound on
	$|e_\rho|=|f_\eta(\rho)-f(\rho)|$,
	\begin{equation}\label{eq:certbound}
		|s_\eta-s|\ \le\ B\coloneq \frac{E_1+s_\eta\,E_{1/2}}{f_\eta(\half)-E_{1/2}},
		\qquad\text{valid whenever }f_\eta(\half)>E_{1/2}.
	\end{equation}
	Every quantity on the right is computed in one pass over arrays already in
	memory. The proposition uses only the engine assumptions and strict
	rounding: no local CLT, no lattice-span condition, no regime constant
	$\eps_0$, no saddlepoint expansion, and no Arıkan scale
	(proofs in \Cref{app:cert}).
\end{proposition}

The~rank window survives at the extremes of the lattice and the Chernoff cap $e^{\rho t}$
handles the top of the distribution, while the $D$-window never raises $\mathrm{lo}_\rho$ and lowers
$\mathrm{hi}_\rho$ by at most one lattice step of $\tF$ (\Cref{app:cert}).
Carrying $e_\rho$ as a signed interval rather than the
symmetric $E_\rho$ tightens the result further, the two sides of
\eqref{eq:bracket} being asymmetric.
Measured against exact ground truth on exhaustively enumerable models, the bracket holds throughout.
For~i.i.d.\ marginals $p_i=(0.5,0.3,0.2)$ with $m\in\{6,8,10\}$, the certificate sits within a factor $1.14$--$1.45$ of the leading term $C(s)\eta$ at $\eta\in\{0.1,0.05\}$, closing monotonically with $m$; at finer $\eta$ the ratio leaves this band and eventually grows, but this advice has rank at most~$3$, so the exact engine applies to it (\Cref{rmk:map}).
On~heterogeneous advice ($m=8$, three symbols per coordinate), $B$ falls faster than linearly over the first halvings of $\eta$ from $0.2$ and linearly thereafter, down to $\eta=0.2/128$ with no floor in sight.
Since \eqref{eq:certbound} is a theorem
about exact arithmetic while the arrays are floating point,
it certifies the binning error of the exact estimator.
The~additional floating-point error is controlled by \Cref{lem:stable}.
The~reference implementation does not provide a rigorous end-to-end enclosure, but running the same computation in directed-rounding interval arithmetic (\Cref{sec:stable}) would provide one.

Having \Cref{prop:cert}, \Cref{thm:bin}, \Cref{lem:stable} and
\Cref{lem:fftcost}, we are now set to prove \Cref{thm:precision}:
\begin{proof}[Proof of \Cref{thm:precision}]
	\emph{(i)}
    By the triangle inequality, $|\widehat s-s|\le|s_\eta-s|_{\text{bin}}+|\widehat s-s_\eta|_{\text{fp}}$; \Cref{prop:cert} bounds the first term by the computed certificate $B$, while \Cref{lem:stable} separately bounds the floating-point term. Under the interval-arithmetic route both contributions can be enclosed rigorously end-to-end. No hypothesis beyond the engine assumptions enters.
	\emph{(ii)} The scaling law \eqref{eq:Cbound} is precisely
	\Cref{thm:bin}. For the tracking claim: under \Cref{thm:bin}'s hypotheses,
	every contributing level has $D(t)\le m\eta(1+\eps_m)$ by \Cref{lem:br}
	and the telescoping argument, while the finitely many levels with
	$D(t)=\infty$ are covered by the rank-window and Chernoff candidates of
	\eqref{eq:lohi} and contribute the exponentially small localization
	remainder; hence $E_\rho\le\rho m\eta(1+\eps_m)+O(e^{-\Theta(m)})$, plus for truncated merges the term of \Cref{app:cert}, negligible when $m\eta=O(\log(m/\eta))$ and, under the regularity conditions of \Cref{cor:scaling}, throughout the regime, and
	\eqref{eq:certbound} gives $B\le C(s)\eta(1+\eps_m)(1+o(1))$. A run at
	$\eta$ therefore certifies binning error $\delta$ once $\eta=\Theta(\delta/C(s))$; the
	geometric fallback of \Cref{alg:binned} halves $\eta$ at most
	$O(\log(1/\delta))$ times.
	For~uniformly bounded alphabets, each run costs $\widetilde O(m/\eta)$
	(\Cref{lem:fftcost}), and the final run dominates.
\end{proof}

\begin{remark}[status of the guarantees, and the floors]\label[remark]{rmk:certified}
	Precisely which statement is non-asymptotic.
    \emph{(i)} \Cref{thm:exact} is an exact identity with no discretization error: $s$ is a ratio of finite level sums, and the only numerical error is the floating-point evaluation covered by \Cref{lem:stable}. \enquote{Exact} thus refers to the underlying finite-sum identity, not exact numerical representability.
    \emph{(ii)} \Cref{prop:cert} is non-asymptotic and computable: in exact arithmetic and whenever $f_\eta(\half)>E_{1/2}$, $B$ is a rigorous bound on $|s_\eta-s|$ under the engine assumptions alone.
    \emph{(iii)}
	\Cref{thm:bin} is an asymptotic scaling law with computed constants ---
	the interpretive statement explaining why refining $\eta$ works and where
    $C(s)$ comes from --- and it is the only place where its hypotheses enter.
	As for floors on the reachable precision: the certificate's floor is the
	\emph{atom spacing} of the binned surprisal law. For heterogeneous advice
	(rank $\Theta(m)$, \Cref{sec:exact}) the spacing is $e^{-\Theta(m)}$ and
	this floor is irrelevant --- in practice the floating-point floor
	of \Cref{app:stable} binds first --- while for
	$k$-symbol i.i.d.\ advice the $O(m^{k-1})$ levels lie in a range of length $O(m)$, so their typical spacing is only polynomially small, of order $m^{2-k}$, and the
	certificate stalls once the window $m\eta$ falls below the gaps between them; but that is exactly the finite-rank regime
	where the exact engine applies and no certificate is needed. One caution
	in the opposite direction: at finite $m$ the substitution
	$f(\half)\approx\half m\Hbar_{2/3}$ of \eqref{eq:arikanscale} overstates
	the true $f(\half)$, by ${\approx}20\%$ on the i.i.d.\ model above at $m=8$, and it sits in the
	denominator; the $O(\log(1+\ln N)/f_\eta(\half))$ term of $\eps_m$ is what
	keeps $C(s)\eta(1+\eps_m)$ conservative, and the computed certificate,
	which divides by the computed $f_\eta(\half)$, avoids it entirely.
\end{remark}

\begin{remark}[which engine to use, which guarantee to state: a user's map]\label[remark]{rmk:map}
	\emph{(1) Level structure.} If the surprisals are commensurate or, more
	generally, of bounded rank $d$ --- a model hypothesis, like commensurability
	itself (\Cref{sec:exact}) --- the level count is $O(m^d)$: run the exact
	engine (\Cref{alg:exact}). No certificate is needed there, and for commensurate advice on a span $h\ge2\eta$ the scaling law's lattice
	hypothesis would in fact fail ($\Lameta=0$, \Cref{rmk:sublattice}).
	Otherwise, bin and run \Cref{alg:binned}.
	\emph{(2) Certificate.} Each binned run returns $s_\eta$ with the bound $B$
	of \Cref{prop:cert} on $|s_\eta-s|$, valid under its engine assumptions
	alone; refine $\eta$ until $B\le\delta$. Two floors bind: the atom spacing of
	the binned law ($e^{-\Theta(m)}$ for heterogeneous advice) and, in the
	floating-point realization, the round-off floor of \Cref{app:stable}. The scaling
	law in Eq.~\eqref{eq:Cbound} explains the observed $B=O(C(s)\eta)$ under its own
	hypotheses but is never load-bearing (\Cref{rmk:certified}).
\end{remark}

\section{Cryptographic applications}\label{sec:apps}
We report the information-theoretic exponent $s=f(1)/f(\half)$ of
\Cref{sec:setup} and nothing beyond it: Montanaro's constant and additive
overhead shift the logarithm of the quantum query complexity by $O(1)$, outside the exponent by construction, as
in~\cite{bashiri_super-quadratic_2026}. One rule keeps the classical baseline honest --- \emph{attack seeds or black boxes, not structured keys}. Where the secret is
algebraic (NTRU ternary~\cite{hoffstein_ntru_1998}, ML-KEM binomial~\cite{bos_crystals_2018},
LWE~\cite{regev_lattices_2009}, LPN~\cite{blum_noise-tolerant_2003,alekhnovich_more_2011}, as catalogued
in~\cite{bashiri_super-quadratic_2026}), full-key ordered enumeration is not the optimal classical
attack: decoding can replace it outright~\cite{blum_noise-tolerant_2003}, and lattice reduction
absorbs all but a guessed sub-block of the secret~\cite{matzov_report_2022}, though the
exact efficiency of the latter is contested~\cite{ducas_does_2023}. An exponent
measured against full-key ordered guessing would therefore flatter the quantum
side. We instantiate only where likelihood-ordered enumeration is provably the
optimal classical attack.

\paragraph{(A) Cold-boot advice.} DRAM cells do not lose their contents the
instant power is cut; they decay towards a chip-dependent ground state, so an
attacker who reboots and dumps memory recovers a noisy copy of whatever was
held there~\cite{halderman_lest_2009}. Bits flip independently and asymmetrically,
$\alpha=\PP(0\to1)\ll\beta=\PP(1\to0)$: an exact binary product posterior.
Conditioning on the dumped bit under a uniform prior gives two marginal types: a bit read as $1$ is $1$ with probability $\tfrac{1-\beta}{1-\beta+\alpha}$, and a bit read as $0$ is $1$ with probability $\tfrac{\beta}{\beta+1-\alpha}$.
We~report the \emph{expected} composition, giving the first type to $\lfloor Wf_1\rceil$ of the $W$ coordinates, where $f_1=\tfrac{1-\beta+\alpha}{2}$ is the probability that a bit is read as $1$; the advice is then a function of $(\alpha,\beta,W)$ alone.\footnote{The composition
is the only free convention here. Over the binomial sampling distribution of the
observed bit counts, $s$ lies in $[3.77,4.11]$ for AES-128 at $\beta=0.01$ and
in $[2.70,2.85]$ at $\beta=0.05$ ($95\%$), with the reported value at the mean.}
Two unrelated secrets sit here: the short seeds from which many post-quantum
signatures expand their entire private key~\cite{gauthier-umana_seed_2026}, and
block-cipher keys~\cite{banegas_recovering_2025}. At equal length they carry
identical advice. Neither admits an algebraic shortcut, so ordered enumeration
is the classical attack; and for the ciphers the
corresponding Grover oracles have actually been built and costed, tolerating up
to $40\%$ noise~\cite{gauthier-umana_seed_2026,banegas_recovering_2025}. The cipher rows below follow that work's
channel and targets; its GIFT-128, LowMC-128 and Picnic-L1/L3/L5 instances
are not listed separately because of the shared advice properties. The marginals have rank
$2$, so the finite-rank engine of \Cref{sec:exact} applies and the entries are exact.

\begin{table}[t]\small
	\caption{Computed exponents against the bound of Bashiri, Glaser, May, and Nowakowski~\cite{bashiri_super-quadratic_2026} (BGMN, their Theorem~4.3).
        Block A is exact (finite-rank engine, \protect\Cref{sec:exact}), Block B is exact
		within the synthetic i.i.d.\ Bernoulli models fitted to the reported average residual ranks.
		Block C averages three simulated posteriors, each enclosed to within $\pm0.003$ by \protect\Cref{prop:cert} at bin width $\eta=0.005$.
		For~blocks A and B, every digit shown is certified by ball arithmetic (see \protect\Cref{app:tabcert}).
		\protect\Cref{app:config} lists the configurations of all rows.
		The $^\dagger$ symbol marks values below $2$.}\label{tab:apps}
	\begin{tabular}{@{}llcc@{}}
		\toprule
		setting                    & parameter        & $s$   & BGMN Thm 4.3    \\
		\midrule
		\multicolumn{4}{@{}l}{\emph{A. cold boot, asymmetric channel $\alpha=0.001$, decay rate $\beta$}}\\
		AES-128,                   & $\beta=0.01$     & 3.918 & 2.310           \\
		\quad 128-bit seed         & $\beta=0.05$     & 2.763 & 2.187           \\
		                           & $\beta=0.25$     & 2.199 & 1.975\rlap{$^\dagger$} \\
		AES-256,                   & $\beta=0.01$     & 3.639 & 2.757           \\
		\quad 256-bit seed         & $\beta=0.05$     & 2.713 & 2.387           \\
		                           & $\beta=0.20$     & 2.245 & 2.101           \\
		PRESENT-80                 & $\beta=0.05$     & 2.819 & 1.979\rlap{$^\dagger$} \\
		                           & $\beta=0.20$     & 2.263 & 1.895\rlap{$^\dagger$} \\
		\midrule
		\multicolumn{4}{@{}l}{\emph{B. synthetic Bernoulli models fitted to Keccak SASCA average ranks
        }}\\
		ML-KEM coin (256\,b)       & $2^{69}$         & 2.711 & 2.467           \\
		                           & $2^{128}$        & 2.342 & 2.212           \\
		ML-DSA $\rho'$ (512\,b)    & $2^{69}$         & 3.117 & 2.861           \\
		                           & $2^{128}$        & 2.697 & 2.560           \\
		\midrule
		\multicolumn{4}{@{}l}{\emph{C. AES template attack, $16$ bytes over $256$ values, single trace}}\\
		HW leakage                 & $\mathrm{SNR}=1$ & 2.018 & 1.900\rlap{$^\dagger$} \\
		                           & $\mathrm{SNR}=5$ & 2.032 & 1.904\rlap{$^\dagger$} \\
		\bottomrule
	\end{tabular}
\end{table}

\paragraph{(B) Side-channel analysis of a hash PRF} ML-KEM derives its
encryption coin, and ML-DSA its masking token $\rho'$, by expanding a short
secret through SHAKE-256. Recovering that seed breaks the scheme without
touching its lattice problem, and because the seed is a structureless random
string, enumeration is again the only classical route. The attack
of~\cite{maillard_simulating_2024} is a \emph{soft-analytical} side-channel analysis, and its
two-stage structure is what matters for us. First, a per-bit classifier reads
the power trace and guesses each bit of the Keccak-$f$ state independently;
taken alone this is weak, and its quality is what that work tabulates as an
\emph{accuracy}. Second, what constitutes the actual attack is that those crude
guesses are placed on a factor graph whose nodes are the state bits and whose
edges are the permutation's own XOR, AND and NOT relations, and \emph{belief
propagation} passes them around that graph until they agree. Thousands of
algebraic constraints thereby correct one another, and the marginals that come
\emph{out} are sharper than the classifier that went \emph{in} by many orders
of magnitude. So the tabulated accuracies are the input to the attack, not its
output, and cannot serve as our advice; what the attack reports as its result
is the residual key \emph{rank}, the number of candidates a classical attacker
must still enumerate, averaged over $50$ attack runs, which we use to calibrate a synthetic posterior model.
We~therefore fit a per-bit $\mathrm{Ber}(q)$ advice on which $50$ simulated runs reproduce the reported average in the median.
This~is a synthetic i.i.d.\ surrogate calibrated to that residual rank, not a measured posterior: the true post-propagation marginals are heterogeneous and,
on graphs with cycles, not exactly independent.
\paragraph{(C) Template attacks.} A template attack~\cite{chari_template_2003} is the
strongest form of profiled side-channel analysis. The adversary first
characterizes a device identical to the target, building for each of the $256$
values a key byte may take a statistical model of the leakage it produces, and
then matches a single trace from the target against those models. Since AES
handles its key one byte at a time in the first round, the leakage is modeled as
sixteen independent per-byte posteriors, the standard assumption here and a
product advice of exactly the form our framework assumes, after which enumerating the candidates in likelihood
order is the established classical procedure~\cite{veyrat-charvillon_optimal_2013}, whose guessing
entropy is the standard yardstick for evaluating such attacks~\cite{tanasescu_tight_2021}. It is also the natural
stress test: a real-valued posterior, simulated from a Gaussian Hamming-weight (HW) leakage model at
two illustrative signal-to-noise ratios, $\mathrm{SNR}=1$ and $5$, whose surprisals are not
commensurate, so it must be binned.

\paragraph{What the numbers show.} The exponent is substantially
super-quadratic wherever the advice is skewed --- above $2.7$ across the
practical cold-boot range, and up to $3.12$ in the synthetic model calibrated to
the reported ML-DSA residual rank, an exponent of that model rather than of the unknown belief-propagation posterior --- so our results suggest that there are
cryptographic settings where choosing Montanaro's algorithm over plain Grover is
meaningful.
The prior bound~\cite{bashiri_super-quadratic_2026} is loosest exactly there.
In~the five rows marked~$^\dagger$ it drops below $2$, which certifies nothing at all: its additive penalty $\log(1+\log|K|)$ is
\emph{fixed} while $H_{1/2}$ shrinks under strong leakage, which is also why
the shortest key fares worst. Two smaller effects are worth recording. The
corrected level averages $\Ac_\rho$ of \Cref{lem:rankcdf} raised every entry of
\Cref{tab:apps} above its block-maximal surrogate, so ignoring ties would have
understated the advantage here; the sign is not fixed a priori (\Cref{rmk:ties}).
And the asymmetric channel of block A yields a larger exponent than a
symmetric one at equal $\beta$ ($2.763$ against $2.532$ at $W=m=128$,
$\beta=0.05$), because it is the more skewed of the two.

Block C is the instructive negative.
The~exponent will not move much: \Cref{prop:cert} encloses the exponents of all three simulated traces in $[2.014,2.022]$ at $\mathrm{SNR}=1$.
Uniform advice at this $N=2^{128}$ already
gives $s=2.0027$, so the advice-driven excess is nearer $0.015$ than $0.018$.
A Gaussian-shaped posterior over a
large alphabet is \emph{smooth}, so $\Hbar_{1/2}/\Hbar_{2/3}$ is at most about $1.02$ and
$s\approx2$ across the SNR range we test. What produces a large exponent is
\emph{skew} --- few dominant symbols against a heavy tail --- which
$\mathrm{Ber}(\beta)$ at small $\beta$ has and a smooth $256$-ary posterior
does not.
To~leading order, \eqref{eq:arikanscale} gives $s\approx2\Hbar_{1/2}/\Hbar_{2/3}$: the strength of the leakage sets the size of these two R\'enyi entropies, and their ratio measures how unevenly the remaining uncertainty is spread.
The~ratio equals $1$ for a posterior that is uniform on its support.
A~tail of unlikely symbols raises it, since the order $\half$ weights them more heavily than the order $\tfrac23$.
For~$\mathrm{Ber}(\beta)$ we have $1.55$ at $\beta=0.01$ and growth like $\tfrac23\beta^{-1/6}$ as $\beta\to0$.
In~guessing terms, the classical mean rank is dominated by harder-than-typical keys deep in the order, and the square root in the quantum moment discounts these large ranks the most.
Sharpening a single-trace posterior, by contrast, only concentrates its mass on one Hamming-weight class and drives $s$ back towards $2$.
Combining multiple traces, however, breaks this behavior and creates skew: with ten traces at $\mathrm{SNR}=1$, $s$ already reaches $2.28$ at a classical cost of about $2^{73}$ (mean over three simulations).
All in all, we see that in~our tested settings mostly the shape of the advice, rather than the
strength of the leakage, decides whether quantum guessing is worth deploying.

\section{Conclusion}
For product-distribution advice --- the situation across a broad range of
cryptographic settings --- the Montanaro speedup exponent is not merely
boundable but computable: exactly for finite-rank advice, and otherwise to
any precision above the model's atom-spacing and floating-point floors, with
a certificate each run assembles from its own arrays. The moments'
dependence on the surprisal law reduces the problem to one-dimensional
convolutions; a single family of exponential tilts makes the computation
provably accurate and numerically stable.
The~framework turns the asymptotic \enquote{there is a super-quadratic speedup} into a concrete, per-scheme \enquote{the guessing-moment exponent is $s$}, with bounded discretization error, at the finite parameters cryptology cares about, including settings where the entropy-based estimate of \cite{bashiri_super-quadratic_2026} is still loose.

\paragraph{Code and data availability.}
Our~reference implementation is available at \url{https://github.com/D-VAmpire/quantum-guessing-speedup-calculations}.
This~repository contains our exact and binned engines together with the certificate of \Cref{prop:cert}, the configurations of all rows of \Cref{tab:apps} (see also \Cref{app:config}) with the stored simulated posteriors of its block~C, as well as scripts that reproduce the computed values of this work and the running example with its figure.
\begin{acks}
The authors thank Alexander May, Franz J. Schreiber, Ronald de Wolf, Nicolai Kraus, Alexis Bagia, and Max Zhao for fruitful discussions and valuable feedback that helped bring this work to its final form.
The authors used Anthropic’s Claude tools to assist with literature searches, numerical analyses, proofreading, streamlining and stylistic refinement. OpenAI’s GPT-5.6 Sol was additionally used to review the manuscript. All AI-generated output was critically evaluated by the authors, who take full responsibility for the content and conclusions of this work.
The authors also gratefully acknowledge the following sources of funding. The group of C.\@ Schubert and J.\@ Seifert is supported by the DFG Priority Program SPP 2514, Berlin Quantum, as well as the BMFTR projects PQ-CCA and PraktiQOM. N.\@ Paskarbeit and M.\@ Margraf are also supported by the BMFTR's PQ-CCA project. M.\@ J.\@ Kramer acknowledges support from the DFG Priority Program SPP 2514; the BMFTR projects PraktiQOM, QuSol, and HYBRID++; the Munich Quantum Valley; Berlin Quantum; and the Clusters of Excellence MATH+ and ML4Q.
\end{acks}

\bibliographystyle{ACM-Reference-Format}
\bibliography{refs}
\clearpage
\appendix
\section{The Arıkan scale of \texorpdfstring{$f(\rho)$}{f(ρ)}}\label[appendix]{app:cscale}
For~the likelihood-ordered $G$ and every $\rho>0$, \citeauthor{arikanInequalityGuessingIts1996}'s inequality~\cite{arikanInequalityGuessingIts1996} brackets the guessing moment directly:
\[
	(1+\ln N)^{-\rho}\,Z_\rho^{1+\rho}\ \le\ \E_P[G^\rho]\ \le\ Z_\rho^{1+\rho},
	\qquad Z_\rho\coloneq\sum_xP(x)^{1/(1+\rho)},
\]
and for a product source $\ln Z_\rho^{1+\rho}=\rho\sum_iH_{1/(1+\rho)}(p_i)$, which gives \eqref{eq:arikanscale}.
(The~base-$2$ restatement in \cite[Lem.~2.2]{bashiri_super-quadratic_2026} carries the slightly larger divisor $(1+\log_2|K|)^\rho$.)
For~uniformly bounded alphabets, $\ln N=O(m)$, so the additive slack $\rho\ln(1+\ln N)$ is $O(\log m)$.
If~moreover $\Hbar_{1/(1+\rho)}$ stays bounded away from zero, the main term is of order $m$, and
\[
	f(\rho)=\rho\,m\,\Hbar_{1/(1+\rho)}\left(1+O\left(\frac{\log m}{m}\right)\right).
\]
In~particular, if $\Hbar_{2/3}$ stays bounded away from zero, then so does $\Hbar_{1/2}\ge\Hbar_{2/3}$, and $f(\half)=\half m\Hbar_{2/3}(1+o(1))$, $f(1)=m\Hbar_{1/2}(1+o(1))$ and $s=2\Hbar_{1/2}/\Hbar_{2/3}\,(1+o(1))$.
With~$\Hbar_{2/3}$ bounded away from zero, these $(1+o(1))$ forms need only $\log(1+\ln N)=o(m)$, not bounded alphabets, which is how \Cref{thm:bin} uses \eqref{eq:arikanscale}.
This~is the denominator scale of the scaling law \eqref{eq:Cbound}, and the reason why the orders $\tfrac23$ and $\half$ recur.
\Cref{cor:arikan} sharpens the bracket into an identity with a one-signed remainder.

\begin{proof}[Proof of \Cref{cor:arikan}]
	Write~$\alpha\coloneq1/(1+\rho)$, so that $\phi_\rho=-\alpha$.
	By~the definition of the tilt, $\nu(u)=e^{\alpha u+\Lam(-\alpha)}\,\nu^{(-\alpha)}(u)$; summing over $u\le t$ and inserting this into \eqref{eq:pinu} gives
	\[
		\Fs(t)=e^{\Lam(-\alpha)+\alpha t}\,\mathcal M_\rho(t),\ \
		\pi(t)=N\,e^{\Lam(-\alpha)-(1-\alpha)t}\,\nu^{(-\alpha)}(t).
	\]
	In~$\pi(t)\Fs(t)^\rho$ the exponent of $t$ is $-(1-\alpha)+\rho\alpha=0$.
	Since $e^{\Lam(-\alpha)}=\E_U[e^{-\alpha S}]=N^{-1}\sum_xP(x)^\alpha$ and $(1+\rho)(1-\alpha)=\rho$, the remaining constant is
	\[
		N\,e^{(1+\rho)\Lam(-\alpha)}=N^{-\rho}Z_\rho^{1+\rho}=N^{-\rho}e^{\rho H_\alpha(P)},
	\]
	which proves \eqref{eq:centring}.
	By~\Cref{lem:rankcdf}(iii) and $\Ac_\rho=\kappa_\rho\,(N\cdot\Fs)^\rho$,
	\begin{align*}
		e^{f(\rho)} & =N^\rho\sum_t\pi(t)\,\Fs(t)^\rho\,\kappa_\rho(t)\\
		            & =e^{\rho H_\alpha(P)}\sum_t\nu^{(-\alpha)}(t)\,\kappa_\rho(t)\,\mathcal M_\rho(t)^\rho,
	\end{align*}
	which is the claimed decomposition.
	Since $\kappa_\rho\le1$, $\mathcal M_\rho\le1$ and the weights $\nu^{(-\alpha)}(t)$ sum to $1$, the last sum is at most $1$, so $c_\rho\le0$, which is Arıkan's upper bound.
	The~lower end $c_\rho\ge-\rho\ln(1+\ln N)$ is his lower bound above.
\end{proof}

\section{Proof of the lattice shift sandwich}\label[appendix]{app:sandwich}
\begin{proof}[Proof of \Cref{lem:sandwich}]
	Strict rounding gives $0<\tS-S\le m\eta$ pointwise, so that, for every key $y$ and every $t$, $\tS(y)\le t$ implies $S(y)<t$, and $S(y)\le t$ implies $\tS(y)\le t+m\eta$.
	Hence
	\begin{equation}\label{eq:cdfsandwich}
		\tF(t)\ \le\ \Fs(t^-)\ \le\ \Fs(t)\ \le\ \tF(t+m\eta)\quad\text{for all }t.
	\end{equation}
	At~$t=S(x)$, \eqref{eq:cdfsandwich} places $J(x)$ in $W(x)$.
	For~$\widetilde J(x)$, the inequalities $S(x)<\tS(x)\le S(x)+m\eta$ and the monotonicity of $\tF$ give $\tF(S(x))\le\tF(\tS(x)^-)$ and $\tF(\tS(x))\le\tF(S(x)+m\eta)$, so that $\widetilde J(x)\subseteq W(x)$ as well.
	Now~$G(x)\in J(x)$ by \Cref{lem:rankcdf}(ii), and $\tAc_\rho(\tS(x))$ averages $r^\rho$ over $r\in\widetilde J(x)$; as ranks are positive integers, both $G(x)^\rho$ and $\tAc_\rho(\tS(x))$ thus lie between $\max(N\cdot\tF(S(x)),1)^\rho$ and $(N\cdot\tF(S(x)+m\eta))^\rho$.
	The~floor at $1$ keeps $\Delta$ finite also where $\tF(S(x))=0$, as at the minimum-surprisal key, since $\tF(\min S)\le\Fs(\min S^-)=0$ by \eqref{eq:cdfsandwich}.
	Taking logarithms gives the bound.
\end{proof}
\noindent In~this proof, strict rounding is needed only for the two lower ends: the implication $\tS(y)\le t\Rightarrow S(y)<t$ behind the left inequality of \eqref{eq:cdfsandwich} places the lower end of $J(x)$ in $W(x)$, and $S(x)<\tS(x)$ does the same for $\widetilde J(x)$; the upper ends need only $\tS-S\le m\eta$.

\section{Proof of the local CLT}\label[appendix]{app:lclt}

\paragraph{Reduction to integer summands, and the correct hypothesis.}
Dividing by $\eta$, the binned surprisals are $\widetilde s_i(a)=\eta\,k_{ia}$ with $k_{ia}\in\Zb$, so $\tS/\eta=\sum_ik_{i,X_i}$ is a sum of independent, integer-valued, non-identically distributed variables.
This~is a standard object, whose local CLT takes a hypothesis on the \emph{product} of characteristic functions, frequency by frequency across the band $|t|\le\pi/\eta$.
A~per-coordinate Cram\'er condition is the wrong shape here, and it can fail: the band is sized to the span of the \emph{sum}, while a coordinate whose binned values lie on a coarser lattice of span $k_i\eta$ with $k_i\ge2$, such as any binary coordinate with binned gap above $\eta$, has $|\psi_i|=1$ at its own resonance $t=2\pi/(k_i\eta)$, inside the band.
For~a binary coordinate, $k_i$ is its binned gap divided by $\eta$, so $k_i$ grows as $\eta$ is refined, and the resonance lies ever deeper inside the widening band.
The~exponent $\Lameta$ of \eqref{eq:lameta} is exactly the product hypothesis, computed.
The~argument below is the classical Fourier route for lattice sums (see \cite{petrov_sums_1975} for the framework), with all constants derived in place.
Throughout, $\sigma_i^2$ is the variance of $Y_i$ under $\widetilde U_\tau$ (\Cref{lem:lclt}) and $\beta_i$ its absolute third central moment, so that $\sum_i\sigma_i^2=\tsig_\tau^2$ and $\sum_i\beta_i=\beta_\tau$.

\paragraph{The band is never empty.}
Every random variable $Y$ on $\eta\Zb$ with mean $\mu$ satisfies $\E|Y-\mu|^3\ge\tfrac\eta2\Var(Y)$.
Indeed, at most one support point $y_0$ lies at a distance $d_0<\eta/2$ from $\mu$; say $y_0\le\mu$, with probability $q$.
Every support point above $\mu$ then lies at a distance of at least $\eta-d_0>\eta/2$, and since $Y-\mu$ has mean zero, these points carry $\sum_{y>\mu}\PP(Y=y)\,(y-\mu)\ge q\,d_0$.
Their contribution to $\E[|Y-\mu|^2(|Y-\mu|-\eta/2)]$ is therefore at least $(\eta-d_0)(\tfrac\eta2-d_0)\,q\,d_0\ge q\,d_0^2(\tfrac\eta2-d_0)$, which offsets the contribution of $y_0$, the only negative one.
Applied to each $Y_i$, this gives $\beta_i\ge\tfrac\eta2\sigma_i^2$, hence $\beta_\tau\ge\tfrac\eta2\tsig_\tau^2$ and $\delta_1\le3\tsig_\tau^2/(2\beta_\tau)\le3/\eta<\pi/\eta$.
Thus~the band $\delta_1\le|t|\le\pi/\eta$ of \eqref{eq:lameta} is never empty, the case $\Lameta=+\infty$ in \Cref{lem:lclt} does not occur, and the split of the inversion integral below is well defined.
The~inequality is tight: equality holds for the symmetric law on two adjacent lattice points.

\begin{proof}[Proof of the product bound in \Cref{lem:lclt}]
	For~coordinate $i$, with centered binned values and tilted weights $p^{(\tau)}_{i,a}$, the imaginary parts cancel in pairs, and $1-\cos\theta=2\sin^2(\theta/2)$ gives
	\[
		|\psi_i(t)|^2=\sum_{a,b}p^{(\tau)}_{i,a}p^{(\tau)}_{i,b}\cos\left(t\left(\widetilde s_i(a)-\widetilde s_i(b)\right)\right)=1-\chi_i(t),
	\]
	where $\chi_i(t)\coloneq\sum_{a<b}4p^{(\tau)}_{i,a}p^{(\tau)}_{i,b}\sin^2(t\Delta_{iab}/2)$ is the $i$-th summand of \eqref{eq:lameta}.
	Since $1-x\le e^{-x}$ for all $x$, which also covers $\chi_i(t)=1$, that is, $|\psi_i(t)|=0$, we get
	\[
		\left|\prod_i\psi_i(t)\right|\le e^{-\chi(t)/2},\qquad\chi(t)\coloneq\sum_i\chi_i(t),
	\]
	and $\chi(t)\ge\Lameta$ on the band by \eqref{eq:lameta}.
	We~bound via the surrogate $\chi$ rather than the exact $-\sum_i\ln|\psi_i|^2$ deliberately: the exact quantity is not a trigonometric polynomial, so its minimum is not FFT-certifiable, and its Lipschitz constant blows up where some $|\psi_i|\to0$.
	The~surrogate loses only where individual coordinates are strongly non-resonant, exactly where there is room to spare, and it errs in the safe direction.
\end{proof}

\paragraph{Certifying $\Lameta$.}
Write~$k_{iab}\coloneq\Delta_{iab}/\eta\in\Zb$ for the binned gaps in lattice steps and $u\coloneq t\eta/(2\pi)$, so that the band is $u_1\le|u|\le\half$ with $u_1\coloneq\delta_1\eta/(2\pi)$.
With~$4\sin^2(\pi uk)=2(1-\cos2\pi uk)$, the surrogate $\chi$ is a cosine polynomial of period $1$ in $u$,
\begin{gather*}
	\chi=\sum_{\ell\ge1}c_\ell\left(1-\cos(2\pi u\ell)\right),\\
	c_\ell\coloneq\sum_{(i,a<b):\,k_{iab}=\ell}2p^{(\tau)}_{i,a}p^{(\tau)}_{i,b},
\end{gather*}
whose coefficients are accumulated in $O(\sum_ib_i^2)$ time.
Its~values at the $N_{\mathrm{grid}}$ points $u=j/N_{\mathrm{grid}}$ follow from one FFT of the coefficients folded modulo $N_{\mathrm{grid}}$, in $O(N_{\mathrm{grid}}\log N_{\mathrm{grid}})$ time.
Since the weights $2p^{(\tau)}_{i,a}p^{(\tau)}_{i,b}$ of each coordinate sum to $1-\sum_a(p^{(\tau)}_{i,a})^2\le1$, and each gap $\Delta_{iab}$ is at most the binned range of coordinate $i$, which exceeds its true range by less than $\eta$, the $u$-derivative satisfies
\[
	\left|\frac{d\chi}{du}\right|\le2\pi\sum_{\ell\ge1}\ell\,c_\ell\le\frac{2\pi(R+m\eta)}{\eta},
\]
with $R$ the surprisal range of \Cref{sec:exact}.
Since $\chi$ is even in $u$, it suffices to evaluate it at the grid points inside $[u_1,\half]$ and at the two endpoints.
Every point of the band then lies within $1/N_{\mathrm{grid}}$ of an evaluated point, so the smallest evaluated value $\chi_{\min}$ certifies $\chi_{\min}-2\pi(R+m\eta)/(\eta N_{\mathrm{grid}})\le\Lameta\le\chi_{\min}$.
Under the conditions of \Cref{cor:scaling}, where $R=\Theta(m)$, accuracy $\eps m$ needs $N_{\mathrm{grid}}=\Theta(1/(\eps\eta))$, at a total cost of $\widetilde O(m+1/\eta)$ for fixed $\eps$, below the engine's own cost.
The~hypothesis of \Cref{thm:bin} asks for $\Lameta\ge\lambda m$ over the whole tilt range $[\widetilde\tau_H,0]$, and a grid of tilts certifies this as well.
Since $\partial_\tau p^{(\tau)}_{i,a}=p^{(\tau)}_{i,a}(\widetilde s_i(a)-\E_\tau\widetilde s_i)$, each summand $\chi_i$ changes with $\tau$ at rate at most $4r_i$, where $r_i$ is the binned range of coordinate $i$, so $|\partial\chi/\partial\tau|\le4(R+m\eta)$; and replacing the lower band edge by the $\tau$-independent $3/(2\max_ir_i)\le\delta_1$ (as $\beta_\tau\le\max_ir_i\,\tsig_\tau^2$ and $\sigma_i\le r_i/2$) only widens the band, which can only lower the minimum.
Tilts spaced $\eps m/(4(R+m\eta))$ apart, which is $\Theta(\eps)$ under \Cref{cor:scaling}, therefore certify a lower bound on $\Lameta$ over the whole tilt range to accuracy $2\eps m$, at $O(1/\eps)$ times the cost above.
In~practice, the candidate minima are the resonances $u=j/k_{iab}$; one evaluates those exactly and lets the grid certify that nothing lower lies between them.

\begin{proof}[Proof of \Cref{lem:lclt}]
	Let~$\psi=\prod_i\psi_i$.
	The~span-$\eta$ inversion formula and the Gaussian inversion are
	\begin{gather*}
		p_k=\frac{\eta}{2\pi}\int_{-\pi/\eta}^{\pi/\eta}e^{-its_k}\psi(t)\,dt,\\
		\frac{1}{\tsig_\tau}\varphi(s_k/\tsig_\tau)=\frac{1}{2\pi}\int_{\Rb}e^{-its_k}e^{-\tsig_\tau^2t^2/2}\,dt .
	\end{gather*}
	Subtracting the two formulas and using $|e^{-its_k}|=1$,
	\[
		\left|g_k-\frac1{\tsig_\tau}\varphi(s_k/\tsig_\tau)\right|\le\frac{J_1+J_2+J_3}{2\pi},
	\]
	where $J_1=\int_{|t|\le\delta_1}|\psi-e^{-\tsig_\tau^2t^2/2}|\,dt$ collects the low frequencies, $J_2=\int_{\delta_1<|t|\le\pi/\eta}|\psi|\,dt$ the high frequencies, and $J_3=\int_{|t|>\delta_1}e^{-\tsig_\tau^2t^2/2}\,dt$ the Gaussian tail.
	The~split point $\delta_1$ of \Cref{lem:lclt} lies below $\pi/\eta$ (see above); the bound $\delta_1\le3\tsig_\tau^2/(2\beta_\tau)$ is the Taylor radius used below, and $\delta_1\le\sqrt3/(2\max_i\sigma_i)$ keeps $\max_i\sigma_i^2t^2\le\tfrac34$, inside the branch radius of the logarithm.

	\emph{Low frequency $J_1$ (Edgeworth, with the quartic term).}
	With~$w_i(t)\coloneq\psi_i(t)-1$, Taylor expansion of the characteristic function to second and third order gives $|w_i|\le\sigma_i^2t^2/2$ and $|\psi_i-1+\sigma_i^2t^2/2|\le\beta_i|t|^3/6$.
	For~$|t|\le\delta_1$, the bound $\delta_1\le\sqrt3/(2\max_i\sigma_i)$ yields the branch condition $\max_i|w_i(t)|\le\max_i\sigma_i^2t^2/2\le\tfrac38<\half$, so the principal logarithm satisfies $|\ln\psi_i-w_i|\le|w_i|^2$.
	Summing over the coordinates,
	\[
		\ln\psi(t)=-\tfrac{\tsig_\tau^2t^2}2+R,\qquad
		|R|\le\tfrac{\beta_\tau|t|^3}6+Q(t),\quad
		Q(t)\coloneq\sum_i|w_i(t)|^2 .
	\]
	The~cubic-only bound $|R|\le\beta_\tau|t|^3/6$ that a summed \enquote{Taylor with remainder} would suggest is \emph{false} in general, since it conflates the expansion of $\psi_i$ with that of $\ln\psi_i$.
	On~i.i.d.\ Bernoulli$(0.8)$ advice at $\tau=-1$, the measured ratio $|R|/(\beta_\tau|t|^3/6)$ reaches $1.08$ for every $m$, since this ratio and $\delta_1$ do not depend on $m$ for i.i.d.\ advice; so the quartic term $Q$ is carried, not hedged.
	The~bound $\delta_1\le3\tsig_\tau^2/(2\beta_\tau)$ gives $\beta_\tau|t|^3/6\le\tsig_\tau^2t^2/4$, whence $|e^R-1|\le|R|e^{|R|}$ yields
	\[
		\left|\psi-e^{-\tsig_\tau^2t^2/2}\right|\le\left(\frac{\beta_\tau|t|^3}6+Q(t)\right)e^{Q(t)}\,e^{-\tsig_\tau^2t^2/4}.
	\]
	Hence $\tfrac{\tsig_\tau}{2\pi}J_1\le\tfrac4{3\pi}\tfrac{\beta_\tau}{\tsig_\tau^3}+q_\tau$ with the \emph{quartic correction}
	\[
		q_\tau\coloneq\frac{\tsig_\tau}{2\pi}\int_{|t|\le\delta_1}\left[\frac{\beta_\tau|t|^3}6\left(e^{Q(t)}-1\right)+Q(t)\,e^{Q(t)}\right]e^{-\tsig_\tau^2t^2/4}\,dt,
	\]
	a computed integral, since $Q$ is an explicit trigonometric expression in the marginals.
	By~$|w_i|\le\sigma_i^2t^2/2$ and $\delta_1\le\sqrt3/(2\max_i\sigma_i)$, we have $Q(t)\le\max_i\sigma_i^2\,\tsig_\tau^2t^4/4\le3\tsig_\tau^2t^2/16$ on $|t|\le\delta_1$, so a factor $e^{-\tsig_\tau^2t^2/16}$ survives in the integrand and $q_\tau=O(\max_i\sigma_i^2/\tsig_\tau^2)$.
	This~is $O(1/m)$ under the regularity conditions of \Cref{cor:scaling}, below the cubic term $\Theta(m^{-1/2})$ by only a factor $O(m^{-1/2})$, so an explicit $\eps_\tau$ must include it.

	\emph{High frequency $J_2$ (lattice span).}
	By~the product bound, $|\psi(t)|\le e^{-\Lameta/2}$ on $\delta_1\le|t|\le\pi/\eta$, so $J_2\le\tfrac{2\pi}\eta e^{-\Lameta/2}$ and $\tfrac{\tsig_\tau}{2\pi}J_2\le\tfrac{\tsig_\tau}\eta e^{-\Lameta/2}$, a computed number, not an order symbol.
	Whether it is small is precisely the content of the hypothesis $\Lameta\ge\lambda m$; that neither $\gcd=1$ nor bounded per-coordinate leakage implies this is shown by the counterexample in \Cref{rmk:sublattice}.

	\emph{Gaussian tail $J_3$.}
	This~is the Gaussian's own integral beyond $\delta_1$, \emph{not} aliasing: the genuine lattice images beyond $\pm\pi/\eta$ are already excluded by truncating the inversion integral to one Nyquist period.
	Exactly, $J_3=\tfrac{2\sqrt{2\pi}}{\tsig_\tau}\,\Gbar(\tsig_\tau\delta_1)$, so $\tfrac{\tsig_\tau}{2\pi}J_3=\sqrt{2/\pi}\,\Gbar(\tsig_\tau\delta_1)$, which is evaluated, not asserted.
	With~the computed $\delta_1$, it is below $10^{-15}$ at $m=256$ and $\eta=0.01$ on the ensembles we tested (i.i.d.\ Bernoulli$(0.7)$, Bernoulli$(0.8)$, Bernoulli$(0.9)$ and ternary $(0.6,0.25,0.15)$ advice, and $228$ Bernoulli$(0.8)$ with $28$ Bernoulli$(0.5025)$ coordinates, each at $\tau\in\{-1,-\half,0\}$), and it is reported rather than absorbed.

	Assembling the three bounds gives the display of \Cref{lem:lclt}.
\end{proof}

\section{Proofs of the reversed-hazard bound and the scaling law}\label[appendix]{app:rate}
We~prove \Cref{lem:br}, explain in \Cref{rmk:cdfbe} why its density-level input cannot be weakened to CDF-level control, and prove \Cref{thm:bin} and \Cref{cor:scaling}.

\begin{proof}[Proof of \Cref{lem:br}]
	Fix~$t\ge H$ in the saddle range, let $\tau=\widetilde\tau_t<0$ be the binned saddle at $t$, and put $z=|\tau|\tsig_\tau$.
	Solving the tilt for $\tnu$ level by level gives the exact identity
	\[
		\frac{\tnu(t)}{\tF(t)}=\frac{\tnu^{(\tau)}(t)}{\sum_{j\ge0}e^{\tau j\eta}\,\tnu^{(\tau)}(t-j\eta)} .
	\]
	The~ratio is well defined, since $H=\tLam'(\widetilde\tau_H)$ is a mean of $\tS$, so that some key has $\tS\le H$ and $N\cdot\tF(t)\ge N\cdot\tF(H)\ge1$ for all $t\ge H$.
	This~uses the binned saddle $\widetilde\tau_H$, which every statement over the tilt range $[\widetilde\tau_H,0]$ presupposes.
	It~exists if and only if $\min\tS<H$, since $\tLam'$ increases continuously from $\min\tS$, its limit as $\tau\to-\infty$, to $\tLam'(0)=\E_U[\tS]>\mu_U>H$.
	As~$\tS\le S+m\eta$, a sufficient condition is $H-\min S=\sum_i(H_1(p_i)-\min_as_i(a))>m\eta$.
	For~uniformly bounded alphabets, it holds under the entropy gap of \Cref{thm:bin} once $m$ is large in the regime $\sqrt m\,\eta\le\eps_0$: by compactness, a constant fraction of the coordinates have $H_{1/(1+\rho)}(p_i)-H_1(p_i)\ge g/2$ and hence $H_1(p_i)-\min_as_i(a)$ bounded away from zero, so that $H-\min S=\Theta(m)$.

	By~\Cref{lem:lclt}, $\tnu^{(\tau)}(y)=(\eta/\tsig_\tau)\,(\varphi(\bar y/\tsig_\tau)+r(y))$ with $|r|\le\eps_\tau$, where $\bar y$ is the centered argument, so that $\bar t=0$ at the saddle.
	The~weight $e^{\tau j\eta}$ confines the denominator to displacements $j\eta=O(1/|\tau|)$, across which $\varphi(\bar y/\tsig_\tau)\ge\varphi(0)(1-O(z^{-2}))$, so there the point-mass control is relative, of size $\sqrt{2\pi}\,\eps_\tau(1+o(1))$, as $\sqrt{2\pi}=1/\varphi(0)$.
	Hence the denominator is $\eta\,\varphi(0)\,(1+O(\sqrt{2\pi}\,\eps_\tau)+O(z^{-2}))/(\tsig_\tau(1-e^{-|\tau|\eta}))$, and with $1-e^{-u}\le u$,
	\begin{align*}
		\frac{\tnu(t)}{\tF(t)}&\le\left(1-e^{-|\tau|\eta}\right)\left(1+2\sqrt{2\pi}\,\eps_\tau\left(1+O(\eps_\tau)\right)+O(z^{-2})\right)\\
		&\le|\tau|\,\eta\left(1+\eps_r(t)\right),
	\end{align*}
	where the $O(|\tau|\eta)$ term of $\eps_r$ absorbs the second-order geometric correction.
	The~binned saddle differs from $\tau_t$ by a computed shift: $\tLam'(\tau_t)-t$ consists of the mean shift $\E[\tS-S]\in(0,m\eta]$ and a change-of-measure term $O(|\tau_t|\sigma_t\,\mathrm{sd}(\tS-S))$ with $\mathrm{sd}(\tS-S)=O(\eta\sqrt m)$, so dividing by $\Lam''(\tau_t)=\sigma_t^2$ gives $|\widetilde\tau_t-\tau_t|=O(m\eta/\sigma_t^2+\eta\sqrt m/\sigma_t)$, which is $O(\eta)$ under \Cref{cor:scaling}.
	Replacing $|\tau|$ by $|\tau_t|$ in the bound costs the term $O(|\widetilde\tau_t-\tau_t|/|\tau_t|)$ of $\eps_r$.

	On~the rest of $\{t\ge H\}$, that is, for $|\tau_t|<\tau_0$ and for $t\ge\mu_U$ (where $\tau_t$ is read as $0$), tilt at the binned saddle where it exists and at $\tau=0$ beyond the binned mean $\E_U[\tS]$, reading $\widetilde\tau_t$ and $z$ as $0$ there.
	For~$z\ge1$, write the single-step ratio through the Gaussian hazard $\lambda_\Phi\coloneq\varphi/\Gbar$: with $\varphi$ in place of the point masses, the denominator is a left Riemann sum of the decreasing function $u\mapsto e^{-zu}\varphi(u)$ at the nodes $u=j\eta/\tsig_\tau$, hence at least $\int_0^\infty e^{-zu}\varphi(u)\,du=\varphi(0)/\lambda_\Phi(z)$, and the point-mass errors of \Cref{lem:lclt} are relative of size $O(\eps_\tau)$, so that the ratio is at most $(\eta/\tsig_\tau)\,\lambda_\Phi(z)\,(1+O(\eps_\tau))$.
	For~$z<1$, restrict the denominator to the displacements $j\eta\le\tsig_\tau$, and beyond the binned mean to the levels in $[\E_U[\tS]-\tsig_\tau,\E_U[\tS]]$; there its model value is bounded below by a constant, so that the ratio is $O(\eta/\tsig_\tau)\,(1+O(\eps_\tau))$.
	Since $0<\lambda_\Phi'<1$, we have $\lambda_\Phi(z)\le\lambda_\Phi(0)+z$ with $\lambda_\Phi(0)=\sqrt{2/\pi}$, and $|\widetilde\tau_t|<\tau_0+|\widetilde\tau_t-\tau_t|$, so in all cases the ratio is at most $(\tau_0+|\widetilde\tau_t-\tau_t|+O(1/\sigma_t))\,\eta\,(1+O(\eps_\tau))$, as claimed.
	The~supremum over $\{t\ge H\}$ is therefore governed by the saddle range, where $|\tau_t|\le1$ (\Cref{cor:rate}) closes the statement.
\end{proof}

\begin{remark}[CDF-level control is insufficient]\label[remark]{rmk:cdfbe}
	The~accuracy analysis needs $\ln[\tF(t+m\eta)/\tF(t)]$ for an increment $m\eta=O(\sqrt m)$ in the scaling regime, \emph{sublinear} on the large-deviation scale.
	Cram\'er-level theory controls $\ln\tF$ only up to $o(m)$, an error that can exceed the entire increment, so sharp large deviations of Bahadur--Rao type are needed, and these rest on density-level control.
	Nor~does a CDF-level central limit theorem with a Berry--Esseen rate suffice.
	Take~$\tau=\tau_t$ in the saddle range and $z=|\tau_t|\sigma_t$, which is $\Theta(\sqrt m)$ under the conditions of \Cref{cor:scaling}, and write $F_\tau(y)\coloneq\PP_{U_\tau}[\bar S\le y]$ for the CDF of the centered surprisal under the tilted law.
	Integrating \eqref{eq:tiltid} by parts against $F_\tau$ and replacing $F_\tau$ by its Gaussian approximation $\Phi(\cdot/\sigma_t)$ leaves the boundary term $e^{|\tau|a}(F_\tau(a)-\Phi(a/\sigma_t))$ and an integral term of at most the same size, so a uniform Berry--Esseen bound $\mathrm{BE}_\tau=O(m^{-1/2})$ controls $M(a)$ only up to $2e^{|\tau|a}\mathrm{BE}_\tau$.
	To~leading order in $1/z$, $M(a)$ itself has the Gaussian size $e^{|\tau|a}\varphi(a/\sigma_t)/z$, which is $e^{|\tau|a}/(z\sqrt{2\pi})$ for $a\ll\sigma_t$ and of the same order for $a=O(\sigma_t)$, a range that covers the window $a\le m\eta=O(\sqrt m)$.
	The~error bound is therefore of the same order as $M(a)$ and controls it to no relative accuracy, whereas the ratio $M(m\eta)/M(0)$ needs relative accuracy $o(1)$.
	Hence the density-level statement is genuinely required, and the \emph{lattice} \Cref{lem:lclt}, applied to $\tF$, supplies it.
\end{remark}

\begin{proof}[Proof of \Cref{thm:bin}]
	Split $\E_P[\tAc_\rho(\tS)]$ at the level $H$.

	\emph{On~$\{S\ge H\}$.}
	Here~$N\cdot\tF(S(x))\ge N\cdot\tF(H)\ge1$, as in the proof of \Cref{lem:br}, so the floor in $\Delta$ is inactive and $\tF$ is positive across the window.
	Telescope $\Delta(x)$ in steps of $\eta$ over the lattice points $t_k\ge H$ of its window and put $h_k\coloneq\tnu(t_k)/\tF(t_k)$.
	Since $\tF(t_k-\eta)=\tF(t_k)-\tnu(t_k)$, each increment is $-\ln(1-h_k)\le h_k/(1-h_k)$.
	By~\Cref{lem:br} and \Cref{cor:rate}, every $h_k$ is at most $q\coloneq\eta(1+\eps_r)$, and $q<1$ in the regime $\sqrt m\,\eta\le\eps_0$ as soon as $\eps_0(1+\eps_r)<1$.
	Hence $\Delta(x)\le mq/(1-q)=(1+\eps_m)m\eta$ with $1+\eps_m=(1+\eps_r)/(1-q)$, and \Cref{lem:sandwich} gives
	\begin{multline*}
		e^{-\rho(1+\eps_m)m\eta}\,\E_P[G^\rho;S\ge H]\\
		\le\E_P[\tAc_\rho(\tS);S\ge H]\le e^{\rho(1+\eps_m)m\eta}\,e^{f(\rho)} .
	\end{multline*}

	\emph{On~$\{S<H\}$.}
	The~restriction to $\{S\ge H\}$ is essential, since below $H$ the rate $|\tau_t|$ exceeds $1$ and grows without bound towards $\min S$, and $\Delta$ with it.
	The~region contributes only exponentially little to either moment, though: guessing moments are dominated by \emph{harder}-than-typical keys (the R\'enyi level $\Hbar_{1/(1+\rho)}>\Hbar_1$), while every key below the entropy level has rank at most $e^H$.
	Quantitatively, by~\eqref{eq:rankchernoff} and \eqref{eq:cdfsandwich}, $\tAc_\rho(\tS(x))\le(N\cdot\tF(\tS(x)))^\rho\le e^{\rho\tS(x)}\le e^{\rho(S(x)+m\eta)}$ and $G(x)\le\Gst(x)\le e^{S(x)}$, so both restricted moments are at most $e^{\rho(H+m\eta)}$.
	By~\eqref{eq:arikanscale}, $f(\rho)-\rho H=\rho m(\Hbar_{1/(1+\rho)}-\Hbar_1)+O(\log(1+\ln N))$ is at least $\rho mg-o(m)$ under the entropy gap and $\log(1+\ln N)=o(m)$, while $\rho m\eta=O(\sqrt m)$ in the regime $\sqrt m\,\eta\le\eps_0$, so both contributions are a factor $e^{-\Theta(m)}$ below $e^{f(\rho)}$.
	Combining the two regions, $|f_\eta(\rho)-f(\rho)|\le\rho m\eta(1+\eps_m)$, with the exponentially small remainder absorbed into $\eps_m$.

	\emph{Propagation.}
	Write $f_\eta(\rho)=f(\rho)+e_\rho$.
	Since $f(1)=s\,f(\half)$,
	\[
		s_\eta-s=\frac{f(1)+e_1}{f(\half)+e_{1/2}}-\frac{f(1)}{f(\half)}=\frac{e_1-s\,e_{1/2}}{f_\eta(\half)}\qquad\text{exactly},
	\]
	so $|s_\eta-s|\le m\eta(1+s/2)(1+\eps_m)/f_\eta(\half)$, and the scale $f_\eta(\half)=\half m\Hbar_{2/3}(1+o(1))$ of \eqref{eq:arikanscale} gives the display.
	Here~$\eps_m$ collects $\eps_r\le2\sqrt{2\pi}\,\eps_\tau(1+O(\eps_\tau))+O(\bar z^{-2})+O(\bar\vartheta)+O(\eta)$ from \Cref{lem:br}, with the suprema $\bar z^{-2}\coloneq\sup\{z_t^{-2}:t\ge H,\ |\tau_t|\ge\tau_0\}$ and $\bar\vartheta\coloneq\sup\{|\widetilde\tau_t-\tau_t|/|\tau_t|:t\ge H,\ |\tau_t|\ge\tau_0\}$ over the saddle range; the rest of $\{t\ge H\}$ does not raise $q$, by the second bound of \Cref{lem:br} and the condition that accompanies it.
	At~the edge level, $z_H^{-2}=1/(m\Varone)$, with $\Varone$ the mean per-coordinate surprisal variance under $P$, is of the order of the first Bahadur--Rao correction~\cite{bahadur_deviations_1960}.
	It~further collects the factor $1/(1-q)=1+O(\eta)$, the Arıkan slack of \Cref{cor:arikan}, which enters through this substitution as $O(\log(1+\ln N)/f_\eta(\half))$, and the localization remainder, of relative size $O(e^{-\Theta(m)}/(m\eta))$, which is $e^{-\Theta(m)}$ since $\log(1/\eta)=o(m)$ and is therefore dominated by the slack.
\end{proof}

\begin{proof}[Proof of \Cref{cor:scaling}]
	Let~$r_i\coloneq\max_a\widetilde s_i(a)-\min_a\widetilde s_i(a)\le c+\eta$ be the binned range of coordinate $i$ and $Y_i$ its binned surprisal under $\widetilde U_\tau$, $\tau\in[\widetilde\tau_H,0]$.
	Then~$\sigma_i^2=\Var_\tau(Y_i)\le r_i^2/4$, so $\tsig_\tau^2=\sum_i\sigma_i^2\le\frac14\sum_ir_i^2\le\frac14(c+\eta)(R+m\eta)=O(m)$.
	Each~tilted weight is at least $e^{-|\widetilde\tau_H|(c+\eta)}/b_i$, a constant, so $\sigma_i^2\ge c'r_i^2$, and Cauchy--Schwarz gives $\sum_ir_i^2\ge(R-m\eta)^2/m=\Theta(m)$; hence $\tsig_\tau^2=\Theta(m)$, with no per-coordinate variance floor needed, and the same argument with $\eta=0$ gives $\sigma_t^2=\Theta(m)$.
	Since $|Y_i-\E Y_i|\le r_i\le c+\eta$, we have $\E|Y_i-\E Y_i|^3\le(c+\eta)\sigma_i^2$, so $\beta_\tau\le(c+\eta)\tsig_\tau^2$ and $\beta_\tau/\tsig_\tau^3\le(c+\eta)/\tsig_\tau=O(m^{-1/2})$, uniformly over the tilt range, while $\beta_\tau\ge\sum_i\sigma_i^3\ge\tsig_\tau^3/\sqrt m$ bounds this ratio below by $m^{-1/2}$.
	Since $\beta_\tau\le(c+\eta)\tsig_\tau^2$ and $\sigma_i\le r_i/2$, we have $\delta_1\ge3/(2(c+\eta))$, so the Gaussian-tail term of $\eps_\tau$ is $e^{-\Theta(m)}$; so is the high-frequency term once $\Lameta\ge\lambda m$ and $\log(1/\eta)=o(m)$, and $q_\tau=O(1/m)$ by \Cref{app:lclt}.
	Hence $\eps_\tau=\Theta(m^{-1/2})$.
	In~the bound on $\eps_m$ of \Cref{thm:bin}, the remaining terms are $\bar z^{-2}=O(1/m)$, as $z_t\ge\tau_0\sigma_t$ on the saddle range, $\bar\vartheta=O(\eta)$ by the saddle shift in the proof of \Cref{lem:br}, $O(\eta)=O(m^{-1/2})$ in its regime, and $O(\log m/m)$, so $\eps_m=O(m^{-1/2})$.
	The~condition of \Cref{lem:br} off the saddle range holds for large $m$, as $\eps_\tau=\Theta(m^{-1/2})$, $\sigma_t^2=\Theta(m)$ there too, and the saddle shift is $O(\eta)$.
\end{proof}

\section{Proofs for the computable certificate}\label[appendix]{app:cert}

\begin{proof}[Proof of the bracket \eqref{eq:bracket}]
	Let~$\tS(x)=t$.
	Strict rounding gives $S(x)\in[t-m\eta,t)$.
	Since $\tF$ is a right-continuous step function on $\eta\Zb$ and $m\eta\in\eta\Zb$, the map $u\mapsto\ln[\tF(u+m\eta)/\tF(u)]$ is constant on each cell $[\ell,\ell+\eta)$, and the cells meeting $[t-m\eta,t)$ have $\ell\in\{t-m\eta,\dots,t-\eta\}$, so
	\begin{equation}\label{eq:binloc}
		\Delta(x)\ \le\ \ln\frac{\tF(S(x)+m\eta)}{\tF(S(x))}\ \le\ D(t)
		\qquad\text{whenever }\tS(x)=t.
	\end{equation}
	Three bounds hold for such $x$.
	\emph{(1)~Rank window:} $G(x)\in J(x)\subseteq W(x)$ by \Cref{lem:rankcdf}(ii) and \Cref{lem:sandwich}, and $S(x)\in[t-m\eta,t)$ gives $N\cdot\tF(t-m\eta)<G(x)\le N\cdot\tF(t+m\eta)$; also $G(x)\ge1$.
	\emph{(2)~Chernoff cap:} $G(x)\le N\cdot\Fs(S(x))\le e^{S(x)}<e^t$ by Eq.~\eqref{eq:rankchernoff}.
	\emph{(3)~Window ratio:} if $D(t)<\infty$, \Cref{lem:sandwich} and Eq.~\eqref{eq:binloc} place $G(x)^\rho$ within a factor $e^{\pm\rho D(t)}$ of $\tAc_\rho(t)$.
	Hence $G(x)^\rho\in[\mathrm{lo}_\rho(t),\mathrm{hi}_\rho(t)]$, and multiplying by $P(x)$, summing over the bin, whose $P$-mass is $\tpi(t)$, and then over $t$ gives Eq.~\eqref{eq:bracket}.
\end{proof}

\begin{proof}[Proof of the certificate \eqref{eq:certbound}]
	By~\eqref{eq:bracket}, $E_\rho\coloneq\max\{f_\eta(\rho)-\ln\sum_t\tpi\,\mathrm{lo}_\rho,\ \ln\sum_t\tpi\,\mathrm{hi}_\rho-f_\eta(\rho)\}\ge|e_\rho|$.
	The~exact ratio identity in the proof of \Cref{thm:bin} gives $s_\eta-s=(e_1-s\,e_{1/2})/f_\eta(\half)$, so with $\xi\coloneq|s_\eta-s|$ and $0<s\le s_\eta+\xi$, we get $\xi\,f_\eta(\half)\le E_1+s\,E_{1/2}\le E_1+(s_\eta+\xi)E_{1/2}$, that is, $\xi\,(f_\eta(\half)-E_{1/2})\le E_1+s_\eta E_{1/2}$, which is \eqref{eq:certbound} whenever $f_\eta(\half)>E_{1/2}$.
	The~denominator is the computed lower bound $f_\eta(\half)-E_{1/2}\le f(\half)$, so no Arıkan input enters.
\end{proof}

\paragraph{The $D$-candidates.}
Let~$F$ be any nondecreasing array on the lattice, and form the candidates of \eqref{eq:lohi} from $F$ in place of $\tF$, with $\tAc_\rho(t)$ an average of $r^\rho$ over $r\in(N\cdot F(t-\eta),N\cdot F(t)]$, so that $(N\cdot F(t-\eta))^\rho\le\tAc_\rho(t)\le(N\cdot F(t))^\rho$.
Where $D(t)<\infty$, its term $u=t-m\eta$ gives $\tAc_\rho(t)\,e^{-\rho D(t)}\le(N\cdot F(t-m\eta))^\rho$, and for every $\ell\in\{t-m\eta,\dots,t-\eta\}$ its term $u=\ell$ gives $\tAc_\rho(t)\,e^{\rho D(t)}\ge(N\cdot F(t-\eta))^\rho\,(F(\ell+m\eta)/F(\ell))^\rho\ge(N\cdot F(\ell+m\eta))^\rho$, as $F(\ell)\le F(t-\eta)$.
With~$F=\tF$, the first inequality shows that the $D$-candidate never raises $\mathrm{lo}_\rho$, and the second, at $\ell=t-\eta$, that it lowers $\mathrm{hi}_\rho$ at most to $(N\cdot\tF(t+(m-1)\eta))^\rho$, one lattice step of $\tF$ below the rank window (\Cref{sec:cert}).

\paragraph{Truncated merges.}
Under the truncated merge of \Cref{lem:fftcost}(ii), the two chains, tilted by $\phi\coloneq\phi_\rho<0$ and $\phi+1>0$, return arrays dominated entrywise by the exact binned laws, $\widehat\tnu\le\tnu$ and $\widehat\tpi\le\tpi$ (\Cref{app:stable}).
Their weighted losses $\mathcal D_\nu\coloneq\sum_t(\tnu-\widehat\tnu)(t)\,e^{\phi t}$ and $\mathcal D_\pi\coloneq\sum_t(\tpi-\widehat\tpi)(t)\,e^{(\phi+1)t}$ are nonnegative and known exactly, since $\sum_t\tnu(t)e^{\phi t}=\E_U[e^{\phi\tS}]$ and $\sum_t\tpi(t)e^{(\phi+1)t}=\E_P[e^{(\phi+1)\tS}]$ are products over the coordinates.
For~the count chain, $0\le\tF(t)-\widehat\tF(t)=\sum_{w\le t}(\tnu-\widehat\tnu)(w)\le \mathcal D_\nu\,e^{|\phi|t}$, which gives the enclosure $\tF_{\mathrm{lo}}\le\tF\le\tF_{\mathrm{hi}}$ with $\tF_{\mathrm{lo}}\coloneq\widehat\tF$ and $\tF_{\mathrm{hi}}(t)\coloneq\min\{1,\widehat\tF(t)+\mathcal D_\nu\,e^{|\phi|t}\}$ used in \Cref{prop:cert}: $\tF_{\mathrm{lo}}$ enters $\mathrm{lo}_\rho$ and $\tF_{\mathrm{hi}}$ enters $\mathrm{hi}_\rho$, including the level averages and the windows $D$ formed from them.
The~rank windows stay valid because $\tF_{\mathrm{lo}}\le\tF\le\tF_{\mathrm{hi}}$, the constant $1$ and the Chernoff cap do not involve the array, and the two inequalities above handle the $D$-candidates: for $x$ in bin $t$, the one formed from $\tF_{\mathrm{lo}}$ is at most $(N\cdot\tF_{\mathrm{lo}}(t-m\eta))^\rho\le(N\cdot\tF(t-m\eta))^\rho<G(x)^\rho$, and the one formed from $\tF_{\mathrm{hi}}$ is at least $(N\cdot\tF_{\mathrm{hi}}(\ell+m\eta))^\rho\ge(N\cdot\tF(S(x)+m\eta))^\rho\ge G(x)^\rho$, where $\ell$ is the cell of $S(x)$.
For~the mass chain, $\widehat\tpi\le\tpi$ and $\mathrm{lo}_\rho\ge0$ keep the lower end of \eqref{eq:bracket} valid with $\widehat\tpi$ in place of $\tpi$, and the upper end is restored by
\[
	\sum_t\tpi(t)\,\mathrm{hi}_\rho(t)\ \le\ \sum_t\widehat\tpi(t)\,\mathrm{hi}_\rho(t)+\mathcal D_\pi\max_t\,e^{-(\phi+1)t}\,\mathrm{hi}_\rho(t),
\]
the maximum running over the lattice levels of $\tS$.

\paragraph{Size of the added term.}
Write $\eps_\nu\coloneq \mathcal D_\nu/\E_U[e^{\phi\tS}]$ and $\eps_\pi\coloneq \mathcal D_\pi/\E_P[e^{(\phi+1)\tS}]$ for the relative tilted-mass deficits; both are the $\ell_1$ deficits of truncated merges and hence at most $(m/\eta)^{-A}$, with $A$ the truncation exponent fixed in the proof of \Cref{lem:fftcost}(ii) (\Cref{app:stable}).
Put~$\Gamma(t)\coloneq e^{-(\phi+1)t}\,\mathrm{hi}_\rho(t)$, and let $\bar\Gamma$ be its mean under the tilted mass law, which has weights $\tpi(t)\,e^{(\phi+1)t}/\E_P[e^{(\phi+1)\tS}]$.
The~upper end is then $\E_P[e^{(\phi+1)\tS}]\,\bar\Gamma$, so relative to it the added term is $\eps_\pi$ times the peak-to-mean ratio $\max_t\Gamma/\bar\Gamma$.
For~the peak, $\mathrm{hi}_\rho(t)\le(N\cdot\tF_{\mathrm{hi}}(t+m\eta))^\rho$, and the Chernoff bound $\tF(t')\le e^{\tLam(\phi)-\phi t'}$ together with $\mathcal D_\nu\,e^{|\phi|t'}=\eps_\nu\,e^{\tLam(\phi)-\phi t'}$ gives
\[
	\tF_{\mathrm{hi}}(t')\ \le\ (1+\eps_\nu)\,e^{\tLam(\phi)-\phi t'}.
\]
Since $(\phi+1)+\rho\phi=0$, as $\phi=-1/(1+\rho)$, and $\tLam(\phi)\le\Lam(\phi)$ at $\phi<0$, this yields
\[
	\max_t\Gamma(t)\ \le\ (1+\eps_\nu)^\rho\left(N e^{\tLam(\phi)}\right)^\rho e^{\rho|\phi|m\eta}\ \le\ 2\,Z_\rho^\rho\,e^{\rho m\eta/(1+\rho)},
\]
with $N e^{\Lam(\phi)}=\sum_xP(x)^{1/(1+\rho)}=Z_\rho$ as in \Cref{app:cscale}.
For~the mean, the upper end is at least $e^{f(\rho)}\ge(1+\ln N)^{-\rho}Z_\rho^{1+\rho}$ by Arıkan's lower bound (\Cref{app:cscale}), while $\tS\le S+m\eta$ gives
\[
	\E_P\left[e^{(\phi+1)\tS}\right]\ \le\ e^{(\phi+1)m\eta}\,\E_P\left[e^{(\phi+1)S}\right]\ =\ e^{\rho m\eta/(1+\rho)}Z_\rho,
\]
so that $\bar\Gamma\ge(1+\ln N)^{-\rho}Z_\rho^\rho\,e^{-\rho m\eta/(1+\rho)}$.
Hence the added term is at most $2(m/\eta)^{-A}(1+\ln N)^\rho e^{2\rho m\eta/(1+\rho)}$ relative to the upper end.
When~$\ln N$ is polynomial in $m$, as for uniformly bounded alphabets, it is negligible whenever $m\eta=O(\log(m/\eta))$, once $A$ is large enough; validity does not need this estimate.

\paragraph{The added term in the regime of the scaling law.}
The~exponential factor in this size estimate arises because the peak is bounded at the rank window shifted by $m\eta$, but the mean through $e^{f(\rho)}$, without the shift.
Under the regularity conditions of \Cref{cor:scaling} and in the regime $\sqrt m\,\eta\le\eps_0$ of \Cref{thm:bin}, comparing both with the shifted window removes it for large $m$.
Let~$\widetilde{\mathcal M}_\rho(t')\coloneq\sum_{j\ge0}e^{-|\phi|j\eta}\,\tnu^{(\phi)}(t'-j\eta)$ be the binned analogue of $\mathcal M_\rho$ in \Cref{cor:arikan}, so that $\tF(t')=e^{\tLam(\phi)-\phi t'}\widetilde{\mathcal M}_\rho(t')$, and write $\widetilde\sigma\coloneq\tsig_\phi$, which is $\Theta(\sqrt m)$ (\Cref{app:rate}).
The~bound on $\tF_{\mathrm{hi}}$ above sharpens to $\tF_{\mathrm{hi}}(t')\le e^{\tLam(\phi)-\phi t'}(\widetilde{\mathcal M}_\rho(t')+\eps_\nu)$, and the second inequality of the $D$-candidates gives $\mathrm{hi}_\rho(t)\ge(N\cdot\tF(t+(m-1)\eta))^\rho$ at every level $t$ where the Chernoff cap $e^{\rho t}$ does not bind.
As~$(\phi+1)+\rho\phi=0$ and $\bar\Gamma\ge\tfrac34\min_{t\in T}\Gamma(t)$ for every set $T$ of levels of tilted mass at least $\tfrac34$, the two bounds give
\[
	\frac{\max_t\Gamma(t)}{\bar\Gamma}\ \le\ \frac43\,e^{\rho|\phi|\eta}\left(\frac{\max\widetilde{\mathcal M}_\rho+\eps_\nu}{\min_{t\in T}\widetilde{\mathcal M}_\rho(t+(m-1)\eta)}\right)^\rho
\]
for every such $T$ on which the cap does not bind.
\Cref{lem:lclt} at $\tau=\phi$ bounds every point mass of $\tnu^{(\phi)}$ by $(\eta/\widetilde\sigma)(\varphi(0)+\eps_\phi)$, so that $\max\widetilde{\mathcal M}_\rho\le(\varphi(0)+\eps_\phi)(1+|\phi|\eta)/(|\phi|\widetilde\sigma)$.
The~tilted mass law is the law of $\tS$ under a product law whose per-coordinate weights are those of $\widetilde U_\phi$ up to factors in $(e^{-\eta},e^{\eta})$, so its mean lies within $2\eta\sqrt m\,\widetilde\sigma\le2\eps_0\widetilde\sigma$ of $\tLam'(\phi)$ and its standard deviation is at most $e^{\eta}\widetilde\sigma$.
By~Chebyshev's inequality, the levels within $2e^{\eta}\widetilde\sigma$ of this mean form such a set $T$, and for $t\in T$, the level $t+(m-1)\eta$ lies within $m\eta+2(\eps_0+e^{\eta})\widetilde\sigma=O(\widetilde\sigma)$ of $\tLam'(\phi)$, where \Cref{lem:lclt} bounds $\widetilde{\mathcal M}_\rho$ below by a constant times $1/(|\phi|\widetilde\sigma)$, as $\eps_\phi=\Theta(m^{-1/2})$.
The~cap does not bind on $T$: the Chernoff bound gives $N\cdot\tF(t')\le Z_\rho e^{|\phi|t'}=e^{(1-\alpha)H_\alpha(P)+\alpha t'}$ with $\alpha=1/(1+\rho)$, which is at most $e^{t'-(m-1)\eta}$ once $t'-H_\alpha(P)\ge(m-1)\eta/(1-\alpha)=O(\sqrt m)$, while the convexity of $\alpha\mapsto\ln\sum_xP(x)^\alpha$ places the levels $t'=t+(m-1)\eta$, $t\in T$, at least $\frac{1-\alpha}2\min_{u\in[-1,-\alpha]}\Lam''(u)-O(\widetilde\sigma)=\Theta(m)$ above $H_\alpha(P)$.
Hence $\max_t\Gamma/\bar\Gamma=O(1)$, with a constant that depends only on $\eps_0$ and the constants of \Cref{cor:scaling}, and the added term is $O((m/\eta)^{-A})$ relative to the upper end throughout the regime.

\paragraph{Which hypotheses the certificate uses.}
The~inputs to the two proofs are product structure (so that the engine exists and $\tS$ is a sum of per-coordinate roundings), the positivity convention $p_{i,a}>0$ of \Cref{sec:setup}, strict upward rounding (a definition), \Cref{lem:rankcdf}(ii) and \Cref{lem:sandwich}, both proved by counting, the exact \eqref{eq:rankchernoff}, and arithmetic.
Under truncation, the domination $\widehat\tnu\le\tnu$, $\widehat\tpi\le\tpi$ of \Cref{app:stable} enters as well; Arıkan's lower bound enters only the size estimate of the added term, not its validity.
The~regime constant $\eps_0$ is never used, as no expansion is taken, and no local CLT is invoked.
\emph{Validity is not usefulness}: the certificate stays valid where the scaling law fails, but it is not guaranteed to be informative.
At~Bernoulli$(0.999)$ advice with $m=256$ and $\eta=0.02$, it returns $15.3$, which is rigorous and useless.
This~advice is i.i.d., so it has only $m+1$ levels, and the finite-rank exact engine covers it (\Cref{sec:exact}).
The~correct formulation is \emph{valid under the engine assumptions alone, informative where the level structure is rich enough to measure}.
Where the scaling law's hypotheses fail, as for the near-commensurate ensembles of \Cref{rmk:sublattice}, the certificate stays valid, and any loss of accuracy shows up in $B$ rather than as a silent failure.

\section{Proofs of the engine lemmas}\label[appendix]{app:stable}

\subsection{Optimality and the level averages}\label[appendix]{app:levelavg}

\begin{proof}[Proof of \Cref{lem:rankcdf}(i)]
	Reindexed by rank, the moment is a sum of the terms $r^\rho\,P(G^{-1}(r))$, $r=1,\dots,N$, where the sequence $(P(G^{-1}(r)))_r$ is a permutation of the fixed multiset $\{P(x)\}$.
	Since $(r^\rho)_r$ is strictly increasing, the rearrangement inequality of Hardy, Littlewood and P\'olya shows that the sum is minimal exactly when $r\mapsto P(G^{-1}(r))$ is non-increasing, i.e.\@ when $G$ is likelihood-ordered, and every such $G$ produces the same sequence, hence the same value.
	The~minimizer does not depend on $\rho$: one and the same enumeration is optimal for the classical ($\rho=1$) and the Montanaro ($\rho=\half$) moment.
\end{proof}

\begin{proof}[Proof of the bracket \eqref{eq:Abracket}]
	Write~$b=N\cdot\Fs(t)$, $a=N\cdot\Fs(t^-)$ and $n=b-a\ge1$.
	Every~$r\in J_t$ has $r\le b$, so $\Ac_\rho(t)\le b^\rho$, with equality iff $n=1$.
	Since $x\mapsto x^\rho$ is increasing, $\sum_{r=a+1}^br^\rho\ge\int_a^bx^\rho\,dx$, so that
	\[
		\Ac_\rho(t)\ \ge\ b^\rho\,\frac{1-(1-n/b)^{\rho+1}}{(\rho+1)\,n/b}\ \ge\ \frac{b^\rho}{1+\rho},
	\]
	the last step because $u\mapsto1-(1-u)^{\rho+1}$ is concave on $[0,1]$, vanishes at $0$ and equals $1$ at $1$, hence dominates the chord $u$, here at $u=n/b\in(0,1]$.
\end{proof}

\paragraph{Certified evaluation of $\Ac_\rho$.}
For~$\rho=1$ the block sum is closed-form, $\sum_{r=a+1}^br=[b(b+1)-a(a+1)]/2$.
For~$\rho=\half$, the midpoint rule ($\int_{r-1/2}^{r+1/2}\sqrt x\,dx\le\sqrt r$) and the trapezoid rule ($\int_{r-1}^r\sqrt x\,dx\ge(\sqrt{r-1}+\sqrt r)/2$) for the concave $\sqrt{\cdot}$, summed and telescoped over the block, give the closed-form bracket
\[
	\tfrac23\left[(b+\tfrac12)^{3/2}-(a+\tfrac12)^{3/2}\right]\le\sum_{r=a+1}^b\sqrt r\le\tfrac23\left[b^{3/2}-a^{3/2}\right]+\tfrac{\sqrt b-\sqrt a}{2}
\]
of relative width $O(n^{-3/2})$.
At~target relative precision $\eps$, blocks with $n\le n_\star\coloneq\eps^{-2/3}$ are summed directly, the in-block offsets being integer-exact in IEEE double for $\eps\ge\umach$, and larger blocks use the bracket, whose relative width is then $O(\eps)$.
The~direct sums are formed pairwise, with relative rounding error $O(\umach\log n_\star)$; plain recursive summation would only guarantee $O(n_\star\umach)$, which is $\umach^{1/3}$ at $\eps=\umach$.
Each~level costs $O(n_\star)$, a constant in $m$, both routes stay within the error budget of \Cref{lem:stable}, and carried as an interval, the bracket makes the evaluation certified by construction.
The~ratio $\kappa_\rho(t)=\Ac_\rho(t)/(N\cdot\Fs(t))^\rho$ is formed from $n/b$ via \texttt{expm1}/\texttt{log1p}, never as a difference of close large powers, which would lose about $\log_{10}(b/n)$ digits on thin levels; for instance, $b^{3/2}-a^{3/2}=-b^{3/2}\operatorname{expm1}(\tfrac32\operatorname{log1p}(-n/b))$.
In~this form $\kappa_\rho(t)$ is accurate to a few ulps uniformly in $n/b$, within the budget of \Cref{lem:stable}, while $\rho\ln(N\cdot\Fs(t))$ stays in the log domain.

\begin{remark}[cost of the block-maximal shortcut]\label[remark]{rmk:ties}
	Evaluating the moments at the block-maximal rank $(N\cdot\Fs(t))^\rho$ instead of $\Ac_\rho(t)$ gives exponents $f^\square(\rho)\in[f(\rho),f(\rho)+\ln(1+\rho)]$ by \eqref{eq:Abracket}, so the exact ratio identity in the proof of \Cref{thm:bin} gives, for $s^\square\coloneq f^\square(1)/f^\square(\half)$,
	\[
		|s^\square-s|\ \le\ \frac{\ln2+s\ln\tfrac32}{f(\half)}\ =\ \frac{2\left(\ln2+s\ln\tfrac32\right)}{m\,\Hbar_{2/3}}\,(1+o(1)),
	\]
	the second form by \eqref{eq:arikanscale}, an asymptotic reading rather than a certificate.
	This~bias is $O(1/m)$ for $f(\half)=\Theta(m)$, independently of $\eta$, an additive floor of the kind \Cref{thm:precision} excludes, and its sign is not fixed a priori, since the shortcut inflates both moments.
	Equal-surprisal blocks are the generic structure wherever the exact engine applies (i.i.d.\ advice over $k$ symbols has $\poly(m)$ levels for $k^m$ keys), which is why \eqref{eq:estimator} carries $\tAc_\rho$ throughout.
\end{remark}

\subsection{Merge cost and floating-point stability}

\begin{proof}[Proof of \Cref{lem:fftcost}]
	\emph{(i).} As~in the proof of \Cref{thm:exact}, supports of disjoint blocks add, so the outputs at any one tree level occupy at most $R/\eta+2m$ grid points in total, and sizing each FFT to its output support makes a level cost $O((R/\eta+m)\log(R/\eta+m))$; over $\lceil\log_2m\rceil$ levels this is $\widetilde O(R/\eta+m)$.
	Binning the $B_{\mathrm{in}}$ input symbols and building the $m$ leaves costs $O(B_{\mathrm{in}})$ more, which gives the stated time.
	A~depth-first traversal keeps one pending sibling buffer per level, so the space is dominated by the root vector, $O(R/\eta+m)$, and the cumulative sum is linear.

	\emph{(ii).} Write $\lambda_i$ for the per-coordinate CGF of the chain being merged and $\Lambda_I=\sum_{i\in I}\lambda_i$ for a block $I$ of coordinates; by \Cref{lem:tiltconv}, the exact law at node $I$ is the tilted partial convolution $\mu_I^{(\phi)}$, with analytically known mean $M_I=\Lambda_I'(\phi)$.

	\emph{Per-coordinate curvature is $O(1)$, uniformly in the range.}
	Set~$\gamma_0\coloneq\half$ and $\lambda_0\coloneq\gamma_0/2$, and let $|\xi|\le\lambda_0$.
	Consider first the count chain at a working tilt $\phi=\phi_\rho\in[-\tfrac23,-\half]$, $\rho\in\{\half,1\}$.
	Writing $u_a\coloneq \widetilde s_i(a)-\widetilde s_i(a_{\min})\ge0$ with $a_{\min}$ the minimal-surprisal symbol, its tilted atom weights $\omega_a\propto e^{(\phi+\xi)\widetilde s_i(a)}$ satisfy
	\[
		\omega_a\ \le\ e^{\eta}\,e^{-\gamma u_a}\,\omega_{a_{\min}},\qquad
		\gamma\coloneq-(\phi+\xi),
	\]
	even without the factor $e^{\eta}$.
	The~mass chain needs no separate treatment.
	Unbinned, its tilt at $\phi+1$ is the count chain's tilt at $\phi$ (\Cref{lem:tiltconv}(iii)).
	Binning raises each surprisal by at most $\eta$ and keeps its mass, $e^{-\widetilde s_i(a)}<p_{i,a}\le e^{\eta}e^{-\widetilde s_i(a)}$, so per coordinate the mass chain's tilted atom weights $p_{i,a}\,e^{(\phi+1+\xi)\widetilde s_i(a)}$ are those of the count chain up to factors in $(1,e^{\eta}]$, and the display holds for them with the same $\gamma$.
	Centering at $\widetilde s_i(a_{\min})$,
	\begin{align*}
		\lambda_i''=\Var_\omega(\widetilde s_i)
		 & \le\sum_a u_a^2\,e^{\eta}e^{-\gamma u_a}      \\
		 & \le e^{\eta}\,b_i\sup_{u\ge0}u^2e^{-\gamma u}
		=\frac{4b_i\,e^{\eta-2}}{\gamma^{2}},
	\end{align*}
	and $\gamma\ge\gamma_0-\lambda_0=\gamma_0/2$ on the whole perturbation window, since $-\phi\ge\half=\gamma_0$, so $\lambda_i''\le\bar v_i\coloneq 16b_i\,e^{\eta-2}/\gamma_0^2=O(b_i/\gamma_0^2)$, per coordinate: no uniform alphabet bound enters.
	No~bound on the surprisal range enters either: a symbol of surprisal $\Sigma$ carries tilted weight $e^{-\gamma\Sigma}$.
	This~negative effective tilt is what makes the bound range-independent, and it is exactly what part~(i), which works untilted, lacks.

	\emph{Block concentration.}
	For~$|\xi|\le\lambda_0$ and $V\sim\mu_I^{(\phi)}$, we have the exact identity
	\begin{align*}
		\E\left[e^{\xi(V-M_I)}\right]
		 & =e^{\Lambda_I(\phi+\xi)-\Lambda_I(\phi)-\xi\Lambda_I'(\phi)} \\
		 & =e^{\xi^2\Lambda_I''(\phi+\zeta)/2}\ \le\ e^{\xi^2V_I/2},
		\qquad V_I\coloneq \sum_{i\in I}\bar v_i,
	\end{align*}
	for some $\zeta$ between $0$ and $\xi$ (Taylor--Lagrange on the smooth convex $\Lambda_I$), with $V_I=O(\sum_{i\in I}b_i)=O(B_I)$.
	Markov's inequality and optimizing over $\xi\in[0,\lambda_0]$ give the two-regime bound
	\[
		\PP\left(|V-M_I|\ge u\right)
		\le2\exp\left(-\min\left(\frac{u^2}{2V_I},\ \frac{\lambda_0u}{2}\right)\right),
	\]
	so with $L_\eps\coloneq \ln(2/\eps_{\mathrm{tr}})$ the window $W_I\coloneq [M_I\pm u_I]$, $u_I\coloneq \sqrt{2V_IL_\eps}+2L_\eps/\lambda_0$, carries all but $\eps_{\mathrm{tr}}$ of $\mu_I^{(\phi)}$.
	It~spans $O((\sqrt{V_IL_\eps}+L_\eps)/\eta)=\widetilde{O}(\sqrt{B_I}/\eta)$ grid points, which is the window claimed in part~(ii).

	\emph{Truncated merge and error.}
	Evaluate the tree depth-first: at each node, convolve the two truncated child vectors by an FFT sized to the sum of their supports, then zero every entry outside $W_I$.
	Each~leaf is built directly on its window from its $b_i$ symbols, so that no leaf spans its whole surprisal range.
	Take~$\eps_{\mathrm{tr}}\coloneq (m/\eta)^{-A}/(2m)$, with $A$ a fixed constant that exceeds the polynomial degrees in \Cref{lem:stable} and is large enough that $(m/\eta)^{-A}\le\umach$; for $m/\eta\ge2$, any $A\ge\log_2(1/\umach)$ will do, that is, $A\ge53$ in IEEE double.
	By~induction, the working vector obeys $0\le\widehat\mu_I\le\mu_I^{(\phi)}$ entrywise in exact arithmetic, since convolution of dominated nonnegative vectors is dominated and truncation preserves this, so each truncation discards at most the \emph{true} out-of-window mass, at most $\eps_{\mathrm{tr}}$.
	Since $\|a*b-\widehat a*\widehat b\|_1\le\|a-\widehat a\|_1+\|b-\widehat b\|_1$ for sub-probability vectors, the $\ell_1$ deficits add over the $2m-1$ nodes, $\|\widehat\mu_{\mathrm{root}}-\mu^{(\phi)}\|_1\le(2m-1)\eps_{\mathrm{tr}}\le(m/\eta)^{-A}$.
	This~additive one-sided (mass-deficit) perturbation has size at most $(m/\eta)^{-A}\le\umach$ in every norm, and \Cref{lem:stable} absorbs it.
	Once~$A$ is large enough, the root window contains every entry that contributes to the assembly at tolerance $\poly^{-1}$, and such an entry is at least $\poly^{-1}/L$ (proof of \Cref{lem:stable}), so pointwise the deficit is a relative input error of at most $\umach\poly$.
	In~the cumulative sums of the count chain, the analytic reweighting $e^{|\phi|w}$ of \eqref{eq:logrecover} is increasing in $w$, so the discarded tail below a contributing level $t$ contributes at most $(m/\eta)^{-A}e^{|\phi|t}$ against a retained sum of at least $e^{|\phi|t}\poly^{-1}/L$, again a relative $\umach\poly$.
	(The~mass chain enters the assembly only pointwise.)

	\emph{The $d$-dimensional merge.}
	The~finite-rank engine of \Cref{sec:exact} mirrors this construction with $\Zb^d$ in place of the grid: per-coordinate supports add along each axis, so the nodes of any one round jointly hold $O(m^d)$ cells, each $d$-dimensional FFT is sized to the actual support box of its output, and the merge costs $\widetilde O(m^d)$ in total; the truncation and domination arguments above apply per axis unchanged.

	\emph{Cost.}
	With~$L_\eps=O(\ln(m/\eta))$, node $I$ holds $n_I=O((\sqrt{V_IL_\eps}+L_\eps)/\eta)$ points.
	Over~the $2^j$ nodes of level $j$, Cauchy--Schwarz gives $\sum_I\sqrt{V_I}\le\sqrt{2^j\sum_I V_I}=O(\sqrt{2^jB_{\mathrm{in}}})$, so the level costs $\widetilde O((\sqrt{2^jB_{\mathrm{in}}}+2^jL_\eps)/\eta)$, geometric in $j$ and dominated by the leaf level.
	In~total this is $\widetilde O((\sqrt{mB_{\mathrm{in}}}+m)/\eta)$, which is $\widetilde O(m/\eta)$ for bounded alphabets, with the alphabet sizes entering only through $B_{\mathrm{in}}$.
	The~depth-first traversal holds one buffer per level, for $\widetilde O(\sqrt{B_{\mathrm{in}}}/\eta)$ space in total, and the final cumulative sum is linear in the root support $\widetilde O(\sqrt{B_{\mathrm{in}}}/\eta)$.
\end{proof}

\begin{proof}[Proof of \Cref{lem:stable}]
	We~prove the floating-point bound and the absence of underflow, arguing for \Cref{alg:binned} (binned, two chains); \Cref{alg:exact} is the special case with the single-chain collapse \eqref{eq:pinu}, and it is strictly better conditioned.
	Throughout, $L=\widetilde O(\sqrt{B_{\mathrm{in}}}/\eta)$ is the truncated transform length (\Cref{lem:fftcost}(ii)), and $\poly$ abbreviates a fixed polynomial in $m,1/\eta$.

	\emph{Per-convolution error.}
	For~radix-2 FFTs whose weights carry normwise error $O(\umach)$ (hypothesis~(24.2) of \cite{higham_accuracy_2002}; standard libraries satisfy this), \citeauthor{higham_accuracy_2002}'s single-transform bound~\cite[Thm.~24.2]{higham_accuracy_2002} gives each FFT-based convolution $a*b=\mathrm{IFFT}(\mathrm{FFT}(a)\odot\mathrm{FFT}(b))$ of two sub-probability vectors an absolute $\ell_2$ error $O(\umach\log L)$.

	\emph{Merge-tree composition.}
	The~$m$ leaves combine by $m-1$ convolutions in a balanced tree of depth $\lceil\log_2m\rceil$.
	Young's inequality $\|e*\mu\|_2\le\|e\|_2\|\mu\|_1$ passes an error $e$ inherited from one child on to the parent without growth, as the other child $\mu$ is a sub-probability vector, and truncation does not increase it either.
	To~first order, the fresh errors of all $m-1$ internal nodes therefore add up at the root, which carries an absolute $\ell_2$ error $O(\umach\,m\log L)$.
	As~the exact tilted law carries mass at least $1-(m/\eta)^{-A}$ on the $L$ points of the root window, $\|\tnu^{(\phi)}\|_2\ge(1-(m/\eta)^{-A})/\sqrt L$, so
	\[
		\left\|\mathrm{fl}(\tnu^{(\phi)})-\tnu^{(\phi)}\right\|_2\ \le\ \eps_{\mathrm{tree}}\,\|\tnu^{(\phi)}\|_2
	\]
	with $\eps_{\mathrm{tree}}=O(\umach\,m\sqrt L\log L)=O(\umach\poly)$.
	The~same count, over the same $\Theta(m)$ internal nodes, covers the $d$-dimensional merge of the finite-rank engine (\Cref{sec:exact}).

	\emph{Normwise $\to$ entrywise on the contributing entries.}
	The~same mass bound gives the peak $\|\tnu^{(\phi)}\|_\infty\ge(1-(m/\eta)^{-A})/L$.
	The~entries that contribute to the assembly at tolerance $\poly^{-1}$ are at least $\poly^{-1}$ times this peak, so their relative error is at most $\eps_{\mathrm{tree}}\poly$ times $\|\tnu^{(\phi)}\|_2/\|\tnu^{(\phi)}\|_\infty\le\|\tnu^{(\phi)}\|_\infty^{-1/2}=O(\sqrt L)$, that is, $O(\umach\poly)$.
	This~is the one nonstandard step: it works because the contributing entries are within a $\poly$ factor of the $\ell_2$ norm.

	\emph{Log-domain assembly.}
	By~\eqref{eq:logrecover} we assemble $\ln\tF(t)=\tLam(\phi)+\LSE_{w\le t}[-\phi w+\ln\tnu^{(\phi)}(w)]$.
	Computed entries that FFT round-off leaves non-positive are set to zero first, which does not increase their error, as the exact entries are nonnegative; they enter the $\LSE$ as $-\infty$.
	Each~$\LSE$ is thus the logarithm of a sum of nonnegative terms, and we bound its error in the linear domain rather than through the logarithms of single entries.
	For~a contributing level $t$, the sum $\sum_{w\le t}e^{|\phi|w}\,\tnu^{(\phi)}(w)$ has absolute error at most $e^{|\phi|t}$ times the $\ell_1$ error of the root vector, hence at most $e^{|\phi|t}\sqrt L\,\eps_{\mathrm{tree}}\|\tnu^{(\phi)}\|_2$, while it is at least its last term $e^{|\phi|t}\,\tnu^{(\phi)}(t)$, a contributing entry; by the entrywise step, its relative error is $O(\umach\poly)$, and so is the absolute error of its logarithm.
	With~the analytic inputs $-\phi w$ and $\tLam(\phi)$ accurate to $O(m\umach)$, $\ln\tF(t)$ has absolute error $O(\umach\poly)$ at every contributing level, and so has $\ln\tpi(t)$, whose entries enter pointwise.
	The~final $\LSE$ over the levels, bounded the same way in the linear domain and with the level averages of \Cref{app:levelavg}, gives $f_\eta(\rho)$ with absolute error $O(\umach\poly)$.

	\emph{Final ratio.}
	Write $\widehat f_\eta(\rho)=f_\eta(\rho)+e_\rho$ with $|e_\rho|=O(\umach\poly)$.
	Since $s_\eta=f_\eta(1)/f_\eta(\half)$, we have $\widehat s-s_\eta=(e_1-s_\eta e_{1/2})/(f_\eta(\half)+e_{1/2})$.
	Provided $|e_{1/2}|\le f_\eta(\half)/2$, this gives $|\widehat s-s_\eta|=O(\umach\poly\,(1+s_\eta)/f_\eta(\half))$, which simplifies to $O(\umach\poly/m)$ for families with $f_\eta(\half)=\Theta(m)$ and $s_\eta=O(1)$.

	\emph{No underflow.}
	No~contributing value underflows for any practical $m,1/\eta$: the contributing tilted entries are at least $\poly^{-1}(1-(m/\eta)^{-A})/L$, far above the IEEE double underflow threshold; the other contributing floats are log-domain quantities of magnitude $O(m)$; and the rare-event magnitudes $e^{-\Theta(m)}$ enter only through their logarithms, via the analytic factor $-\phi w+\tLam(\phi)$, and are never exponentiated.
\end{proof}

\paragraph{The explicit error budget, and the floor that binds.}
Every step above has explicit constants available: Higham's per-FFT bound, the first-order tree composition over the $m-1$ internal nodes, the $\ell_2$-to-entrywise conversion via the peak bound $\|\tnu^{(\phi)}\|_\infty\ge(1-(m/\eta)^{-A})/L$, the linear-domain bound on the log-domain assembly (including the \texttt{expm1}/\texttt{log1p} evaluation of the level averages, \Cref{app:levelavg}), and the final division by $f_\eta(\half)$.
The~bound of \Cref{lem:stable} on the floating-point error of $\widehat s$ grows as $\eta$ shrinks, through the factor $L\log L$ with $L\propto1/\eta$ in the proofs above, while the certified bound $C(s)\eta$ falls, so IEEE double limits the targets that refining $\eta$ can reach.
This~floating-point floor is the one that binds for heterogeneous advice, whose atom spacing is $e^{-\Theta(m)}$; for low-rank advice, the level gaps bind first, once the window $m\eta$ falls below them (\Cref{rmk:certified}).
Rigorous end-to-end certification requires the interval route at higher working precision, with bit complexity scaling accordingly.

\subsection{Certified values for \texorpdfstring{\Cref{tab:apps}}{Table 2}}\label[appendix]{app:tabcert}
The~rows of blocks A and B of \Cref{tab:apps} have at most $16640$ levels each, few enough to certify their values directly rather than through the floating-point pipeline of \Cref{sec:stable}.
Their advice consists of at most two types of binary coordinates with rational probabilities, so all key probabilities are rationals with one common denominator, and two keys share a level exactly when their numerators coincide.
Big-integer arithmetic therefore identifies, orders and counts the levels exactly, and with them the rank blocks $J_t$ of \Cref{lem:rankcdf}.
For~$\rho=1$ the block sums are integers by the closed form of \Cref{app:levelavg}, so $\E_P[G]$ is an exact rational.
For~$\rho=\half$ we evaluate each block sum in ball arithmetic~\cite{johansson_arb_2017} through the Hurwitz zeta function $\zeta(\sigma,x)=\sum_{k\ge0}(k+x)^{-\sigma}$, continued analytically to $\sigma=-\half$,
\[
	\sum_{r=a+1}^{b}\sqrt r=\zeta\left(-\tfrac12,a+1\right)-\zeta\left(-\tfrac12,b+1\right),
\]
summing short blocks directly.
At~$3000$ bits of working precision, the resulting enclosures of $s$ and of the bound of~\cite{bashiri_super-quadratic_2026} have radii below $10^{-800}$ and none of them straddles a rounding boundary, so every value shown in blocks A and B is the exact value correctly rounded to three decimals.
In~block B the $q$ fitted by Monte Carlo enters exactly as stored, a double-precision binary fraction.
All~twelve rows were computed in under a minute on standard notebook hardware.

\section{Configurations for Table~\texorpdfstring{\ref{tab:apps}}{2}}\label[appendix]{app:config}
This~appendix gives the inputs of \Cref{tab:apps} that \Cref{sec:apps} leaves implicit; our GitHub repository holds the same values.

\paragraph{Block A\@.}
With~$\alpha=0.001$, the composition $\lfloor Wf_1\rceil$ assigns the first type, bits read as $1$, to $63$, $61$, $48$ (AES-128, $\beta=0.01$, $0.05$, $0.25$), $127$, $122$, $103$ (AES-256, $\beta=0.01$, $0.05$, $0.20$) and $38$, $32$ (PRESENT-80, $\beta=0.05$, $0.20$) of the $W$ coordinates.
The~symmetric comparison ($2.532$) uses the channel $\alpha=\beta=0.05$, which gives every bit the marginal $\mathrm{Ber}(0.95)$.

\paragraph{Block B\@.}
The~fitted parameters are $q=0.97072$ and $0.91677$ for the $256$-bit coin at the reported average ranks $2^{69}$ and $2^{128}$, and $q=0.98887$ and $0.97129$ for the $512$-bit $\rho'$.
Each~simulated run draws the seed from the advice and records its rank, breaking ties uniformly at random; bisection sets $q$ so that, over $4\cdot10^5$ simulated experiments, the median of the mean rank of $50$ runs equals the reported average.
The~Monte Carlo standard error of $s$ is below $10^{-4}$; the table uses the stored double-precision $q$ (\Cref{app:tabcert}).

\paragraph{Block C\@.}
Each~of the $16$ key bytes leaks the Hamming weight of the AES S-box output on the key byte XOR a uniformly random plaintext byte, plus Gaussian noise of variance $2/\mathrm{SNR}$, where $2$ is the variance of that Hamming weight.
Its~posterior is the normalized likelihood of the $256$ hypotheses, floored at $10^{-10}$ and renormalized.
For~each SNR, we simulate three such posteriors and evaluate each with the binned engine, rounding to nearest (admissible by the footnote in \Cref{sec:binning}) with Richardson extrapolation $2s_\eta-s_{2\eta}$, and with $\eta$ halved from $0.02$ until successive extrapolants agree to $10^{-6}$; the table shows the means of the three values and of the three bounds.
The~ten-trace example of \Cref{sec:apps} sums the log-likelihoods of ten leakages per byte, each with a fresh plaintext byte.

\paragraph{BGMN column.}
We~evaluate the bound of \cite[Thm.~4.3]{bashiri_super-quadratic_2026} in nats as $2(H_{1/2}(P)-\ln(1+\log_2|K|))/H_{2/3}(P)$, keeping the inner logarithm binary as in their base-$2$ statement.

\end{document}